\documentclass[pra, twocolumn,superscriptaddress,nobibnotes]{revtex4-2}
\usepackage{graphicx}
\usepackage{libertine}
\usepackage{amssymb,amsmath,amsthm}
\usepackage{stmaryrd}
\usepackage{bbold}
\usepackage{bbm}
\usepackage{subcaption}
\usepackage{multirow}
\usepackage{tikz}
\usepackage{tikzsymbols}
\usetikzlibrary{shapes.geometric}
\usetikzlibrary{decorations.pathreplacing}
\usepackage{pgfplots}[compat=1.9]
\usepackage{pgfplotstable}
\pgfplotsset{compat=1.9}
\usepackage{xcolor}
\usepackage{braket}
\usepackage{mathtools}
\usepackage{cases} 
\usepackage{ragged2e}
\usepackage[colorlinks = true, linkcolor=teal, citecolor=teal, urlcolor  = teal]{hyperref}
\usepackage[nameinlink]{cleveref}
\usepackage{natbib}
\usepackage{pifont}
\usepackage{enumerate}
\usepackage{enumitem}

\usepackage{physics}
\usepackage{setspace}
\usepackage{thmtools, thm-restate}

\def\arraystretch{1.3}

\usepackage[most]{tcolorbox}

\newtcbox{\inlinecode}{on line,
  box align=base,
  colback=gray!10, 
  colframe=gray!10, 
  boxsep=0pt,
  left=2pt,
  right=2pt,
  top=1pt,
  bottom=1pt,
  enhanced,
  nobeforeafter,
  fontupper=\ttfamily\scriptsize
}

\usepackage{etoolbox}

\def\typstscale{0.55}

\newcommand{\fu}{Dahlem Center for Complex Quantum Systems, Freie Universit{\"a}t Berlin, 14195 Berlin, Germany}

\newcommand{\fzj}{Institute for Theoretical Nanoelectronics (PGI-2), Forschungszentrum J\"ulich, 52428 J\"ulich, Germany}

\newcommand{\edinburgh}{School of Informatics, The University of Edinburgh, EH8 9AB Edinburgh, Scotland, UK}

\newtheorem{lemma}{Lemma}

\newtheorem{corollary}{Corollary}
\newtheorem{definition}{Definition}

\newcommand{\Pn}{\mathcal{P}^n}
\newcommand{\Cn}{\mathcal{C}^n}
\newcommand{\Sp}{\operatorname{Sp}}

\newcommand{\StabGroup}{\mathcal{S}}

\newcommand{\GL}{\text{GL}}

\newcommand{\unitspace}{\mskip3mu}
\newcommand{\unit}[1]{\unitspace\mathrm{#1}}

\newcommand{\FF}{\mathbbm{F}}
\newcommand{\RR}{\mathbbm{R}}
\newcommand{\id}{\text{id}}

\newcommand{\iu}{\mathrm{i}\mkern1mu}

\crefname{equation}{Eq.}{Eqs.}
\crefname{figure}{Figure}{Figures}
\crefname{section}{Sec.}{Secs.}
\crefname{tabular}{Tab.}{Tabs.}
\crefname{table}{Table}{Tables}
\crefname{lemma}{Lemma}{Lemmata}
\Crefname{theorem}{Theorem}{Theorems}
\crefname{definition}{Definition}{Definitions}
\crefname{appendix}{App.}{Apps.}

\Crefname{equation}{Equation}{Equations}
\Crefname{figure}{Figure}{Figures}
\Crefname{section}{Section}{Sections}
\Crefname{tabular}{Tab.}{Tabs.}
\Crefname{table}{Table}{Tables}
\Crefname{lemma}{Lemma}{Lemmata}
\Crefname{theorem}{Theorem}{Theorems}
\Crefname{definition}{Definition}{Definitions}
\Crefname{appendix}{Appendix}{Appendices}

\begin{document}

\setstretch{1}

\title{\rmfamily\large Hardware-tailored logical Clifford circuits for stabilizer codes }

\author{Eric J. Kuehnke}
\email{eric.kuehnke@fu-berlin.de}
\address{\fu}
\author{Kyano Levi}
\address{\fu}

\author{Joschka Roffe}
\address{\fu}
\address{\edinburgh}

\author{Jens Eisert}
\address{\fu}
\author{Daniel Miller}
\address{\fu}
\address{\fzj}

\date{\today}

\begin{abstract}
    Quantum error correction is the art of protecting fragile quantum information through suitable encoding and active interventions.
    However, it is notoriously difficult to manipulate encoded quantum data without causing uncorrectable errors.
    Here, we introduce a mathematical framework for constructing hardware-tailored circuits that implement any desired Clifford unitary on the logical level of any given stabilizer code. 
    We formulate this task as a discrete optimization problem that minimizes the number of CZ gates while naturally incorporating a parameterization of all Clifford stabilizers.
    Despite the intrinsic difficulty of the circuit compilation problem, our methods scale to at least $16$ physical qubits.
    Because our circuits require very few CZ gates, they provide excellent starting points for constructing fault-tolerant implementations using flag gadgets.
    Using a brute-force search, we develop several such flag gadgets, including weakly fault-tolerant logical Hadamard circuits for the $\llbracket8,3,2\rrbracket$ code that require an unprecedentedly small number of physical two-qubit gates.
    Our work introduces both the mathematics and software for compiling hardware-tailored logical Clifford circuits.
    From a broader perspective, we turn away from the standard generator decomposition and instead focus on the holistic compilation of entire logical circuits, leading to significant savings in practice.
\end{abstract}

\maketitle 

\section{Introduction}
The anticipated advantage of quantum computers has not yet been fully realized because decoherence and operational errors are still severely limiting their performance.
The most promising 
solution to overcome this challenge is presented by \emph{quantum error-correcting codes} (QECCs), in particular, by stabilizer QECCs.
Here, a carefully selected set of Pauli operators (the stabilizer generators) is repeatedly measured, thereby pushing the state of the quantum computer back toward the logical subspace~\cite{gottesman_stabilizer_codes_1997, calderbank_quantum_error_1998}. 
While quantum error correction has been a
well-established theoretical field for many years, it is only recently that emphasis has shifted
toward actualizing core elements of  error correction.
For example, the possibility to ever extend the 
lifetime of a logical qubit by encoding it into more and more physical qubits has been experimentally confirmed~\cite{google_supressing_quantum_2023, google_quantum_error_2025,tham_break_even_2026}.
Moreover, various logical primitives have been implemented in the 
laboratory~\cite{self_protecting_expressive_2024, 
wang_fault_tolerant_2024, 
menendez_implementing_fault_2024, 
bluvstein_logical_quantum_2024,
Butt_Measurement_Free_2026,
dasu_computing_with_2026,
jin_iceberg_beyond_2026,
egan_fault_tolerant_2021, 
erhard_entangling_logical_2021,
postler_demonstration_fault_2022,
burton_genons_double_covers_2024,
ryananderson_high_fidelity_2024,
Bluvstein_2025,
yamamoto_quantum_error_2025,
lacroix_scaling_logic_color_2025}.
With the fundamental principles of error correction thus being established, 
a remaining challenge in making quantum error correction practical is to 
lessen the burden arising from the daunting resource demands of 
logical operations.
This issue has to be addressed and tackled
from several perspectives.
In particular, to this end, sophisticated compilation methods are 
urgently needed, especially  
in the form of methods that take experimental constraints such as limited qubit connectivities into account that are relevant for most 
physical platforms for quantum error correction~\cite{rengaswamy_logical_clifford_2020,lao_fault_tolerant_2020, 
miller_hardware_tailored_2024,
winderl_architecture_aware_synthesis_stabilizer_circuits_2024,Stade_2024,Yuan_full_2025,peham_automated_synthesis_2025}.

This is, of course, not a new problem.
One of the first proposals for universal fault-tolerant quantum computation is \emph{surface code lattice surgery}~\cite{horsman_surface_code_2012}.
While its modular approach and conceptual simplicity offer a clear route to large-scale fault-tolerant quantum computers, 
surface code lattice surgery faces massive resource overheads: 
(i) the stabilizer generators must be measured increasingly often as the code size increases, which slows down computation,
(ii) many of the logical qubits are blocked both to route the flow of data
and to effectively implement logical Clifford gates via parity measurements of logical multi-qubit Pauli operators~\cite{litinski_game_of_surface_codes_2019}, and (iii) it suffers from certain no-go theorems which limit all codes whose stabilizer generators are local in two dimensions~\cite{bravyi_no_go_theorem_2009, bravyi_classification_topologically_2013}.

Recent discoveries of good \emph{quantum low-density parity check} (qLDPC) codes~\cite{bravyi_homological_product_2014,
breuckmann_quantum_low_2021, 
panteleev_asymptotically_good_2022, 
leverrier_decoding_quantum_2023, 
bravyi_high_threshold_2024,
xu_constant_overhead_fault_tolerant_2024} have added an entirely new flavor to the problem. 
They elegantly 
circumvent these no-go results by dropping the locality assumption, motivating a rich and promising research program on generalized lattice surgery for qLDPC codes~\cite{cohen_low_overhead_2022,
cowtan_css_code_2024,
cowtan_ssip_automated_2024,
cross_improved_qLDPC_surgery_2024,
williamson_low_overhead_fault_2024,
stein_architectures_heterogeneous_2024,
swaroop_universal_adapters_quantum_2024,
cowtan_parallel_logical_2025,
he_extractors_qLDPC_2025,
poirson_engineering_csss_urgery_2025,
hillmann_single_shot_2025,
yoder_tour_gross_2025,
gu_qgpu_parallel_logic_2026,
yuan_parsimonious_2026,
webster_pinnacle_architecture_reducing_cost_2026,
cain_shors_algorithm_2026}.
While these various approaches to \emph{qLDPC surgery}, indeed, bring down the required number of qubits, most of them unfortunately inherit several drawbacks from its predecessor for the surface code:
(i) in order to ensure fault tolerance, stabilizer generators must still be measured multiple times, and
(ii) to implement Clifford gates, some logical qubits are blocked.
Although more recent works propose constant-time qLDPC surgery, 
which would solve problem (i),
these approaches come at the cost of a so-called auxiliary complex~\cite{ashikhmin_quantum_data_2020,
baspin_fast_surgery_quantum_ldpc_2025,
xu_batched_high_rate_logical_2025,
cowtan_fast_fault_tolerant_logicalmeasurements_2025,
tan_single_shot_universality_quantum_ldpc_2025,
zhang_time_efficient_logical_2025,
zheng_high_rate_surgery_constant_2025,
chang_constant_time_surgery_2d_hypergraph_2026}.

On top of this, an additional quantum co-processor is required, in order to address the logical qubits in a qLDPC memory, which significantly increases the overall number of qubits.
Moreover, the size and layout of this co-processor depend strongly on the choice of logical Pauli operators whose parities ought to be measured.
This, in turn, considerably limits the flexibility of qLDPC surgery and leads to further overheads for routing logical information.

In a different research line, notions of \emph{code deformation}~\cite{bombin_clifford_gates_2011,vuillot_code_deformation_2019, brown_poking_holes_2017} were extended to certain qLDPC codes~\cite{krishna_fault_tolerant_2021}.
However, code deformation also relies on Pauli parity measurements and 
faces similar challenges as qLDPC surgery.

Cheaper alternatives exist in specific settings, with the most efficient protocols being based on transversal implementations of logical gates
(Originally, the notion of transversality referred to implementations that only act on one qubit per code block at a time~\cite{aharonov_fault_tolerant_1997}.
In the recent literature, however, this notion has been relaxed to include implementations that are only transversal with respect to a specific partition of qubits. Nevertheless, the spread of errors is still contained, which ensures fault tolerance).
Various methods have been proposed for compiling transversal Clifford gates \cite{webster_transversal_diagonal_2023, quintavalle_partitioning_qubits_2023, breuckmann_fold_transversal_2024,sayginel_fault_tolerant_logical_2025,malcolm_computing_efficiently_qldpc_2025,davenport_generalized_2026} as well as non-Clifford gates \cite{lin_transversal_non_clifford_2024, hsin_classifying_logical_2024}.
However, the existence of such gates is not guaranteed for all codes, and design trade-offs are often necessary to provide the necessary structure to support transversal implementations. 

This discussion highlights an apparent gap between state-of-the-art experiments and theoretical ideas:
while  
theory research seeks general and scalable protocols with provable properties,
experimentalists require concrete implementations for specific codes under real hardware constraints.
It can take multiple years of developing, fabricating, and calibrating a quantum device before one can execute an error correction experiment.
Here, early design choices may limit the possibility of migrating to newly-discovered QECCs with better code parameters or logical gates.
Ideally, one would have access to a method that is agnostic to the QECC and, given a target gate and hardware constraints, 
constructs an implementation of the logical gate with as little overhead as possible.
In other words, the task is to decompose a logical operation into a short sequence of physical operations.
In this work, we will refer to such a decomposition as a \emph{circuit implementation} of the logical gate.
However, this problem is notoriously difficult, particularly when it comes to fault-tolerant operations.

In this work, we develop a general framework for the synthesis of efficient circuit implementations of logical Clifford gates.
We substantially improve upon the work of Ref.~\cite{rengaswamy_logical_clifford_2020} by fully characterizing the gauge freedom of circuit implementations of logical gates
and incorporating connectivity constraints~\cite{miller_hardware_tailored_2024}.
Moreover, we minimize the number of  two-qubit gates by translating the problem of logical circuit synthesis into an \emph{integer quadratically constrained program} (IQCP)
\cite{Optimization}.
To facilitate seamless integration with existing software tools, 
we offer our framework as a Python package. 
It is available on GitHub~\cite{htlogicalgates_github}.

For distance-2 error-detecting codes---important testbeds for contemporary experiments~\cite{wang_fault_tolerant_2024, 
menendez_implementing_fault_2024, 
Butt_Measurement_Free_2026,
bluvstein_logical_quantum_2024, 
self_protecting_expressive_2024, 
dasu_computing_with_2026, jin_iceberg_beyond_2026}---we demonstrate how our circuits can achieve weak fault tolerance via an appropriate flag gadget construction.
As a timely application, we design weakly fault-tolerant Hadamard gates for the  ``{smallest interesting color code}''~\cite{campbell_the_smallest_interesting_2016},
which recently has received ample attention due to its suitability for experimental implementation~\cite{bluvstein_logical_quantum_2024, wang_fault_tolerant_2024, menendez_implementing_fault_2024, Butt_Measurement_Free_2026}.
In contrast to previous constructions~\cite{wang_fault_tolerant_2024, Butt_Measurement_Free_2026}, 
our logical Hadamard gates do not rely on teleporting logical qubits into (and back from) a second QECC that admits SWAP-transversal Hadamard gates.
As a consequence, our teleportation-free Hadamard gates enjoy significant resource savings and, therefore, improved performance.

A strength of our method is its flexibility.
It not only applies to generating sets of the Clifford group
but also to arbitrary Clifford elements.
By constructing a single implementation for a sequence of multiple logical gates, we achieve significant savings compared to the naive approach where each logical operation is compiled individually.
In this way, we actualize an idea that has been put forward in Refs.~\cite{chen_tailoring_fault_2025, chen_fault_tolerant_2025}.

\section{Results}

In this paper, we show that the problem of decomposing a logical Clifford operation into physical gates can be formulated as an IQCP.
By adapting and developing further ideas from Ref.~\cite{miller_hardware_tailored_2024}, we can thereby enforce the resulting circuits to respect arbitrary hardware connectivity constraints.
On a high level, we 
introduce an ansatz circuit (parameterized with yet-to-be-determined binary variables)
and impose that it realizes one of the many possible implementations (due to gauge freedom) of the desired logical gate.
The ansatz class is presented in  \Cref{lem:ansatz}, while the gauge freedom is fully parameterized in  \Cref{thm:gauge_freedom}.
In \Cref{crl:optimization}, we identify the relevant equations and formulate an IQCP to solve them.
At the end of this section,  we apply our new algorithm and construct logical gates for various QECCs.

\subsection{Parameterization of our ansatz circuits}

We now introduce notation for our class of ansatz circuits, designed to facilitate the construction of hardware-tailored implementations.
For the mathematical background upon which this work relies, 
see methods section ``Clifford group and stabilizer codes''. 
A \emph{single-qubit Clifford gate layer} (SCL) $U_B$ consists of Clifford gates acting independently on every qubit.
The symplectic matrix $B\in\Sp(\FF_2^{2n})$ that represents such a  fully-transversal $n$-qubit Clifford gate consists of diagonal block matrices $B^{xx}, B^{xz}, B^{zx}, B^{zz}\in\FF_2^{n\times n}$.
Here, the {$i$-th} diagonal entry is given by the symplectic representation of the single-qubit Clifford gate on qubit $i$, e.g., $B^{xx}=\operatorname{diag}(b_1^{xx},\dots,b_n^{xx})$.
A \emph{controlled-$Z$ gate layer} (CZL) $U_G$ consists of CZ gates acting between pairs of qubits.
The symplectic representation of such an $n$-qubit CZL can be characterized by an adjacency matrix $\Gamma\in\mathbbm F_2^{n\times n}$, where the entry $\Gamma_{i,j}$ equals one  if and only if (iff) there is a CZ gate between qubits $i$ and $j$.
The symplectic representation of $U_G$ is given by
\begin{align}
    U_G\longmapsto G=\begin{bmatrix}
        \mathbbm 1 & 0\\
        \Gamma & \mathbbm 1
    \end{bmatrix}  \in \operatorname{Sp}(\FF_2^{2n}) \,. \label{eq:CZL}
\end{align}
Similarly, the qubit connectivity of quantum hardware can be described by means of an adjacency matrix $\Gamma_{\text{con}}$.
This time, $\Gamma_{\text{con},i,j}$ equals one iff qubits $i$ and $j$ are physically connected.
Later, this will allow us to obtain hardware-tailored circuits by imposing $\Gamma \leq \Gamma_\text{con}$, with the inequality understood element-wise~\cite{miller_graphstatevis_interactive_2021}.

With these two types of gate layers, we are now ready to define a class of ansatz circuits.
Our circuits consist of multiple SCLs and CZLs, denoted by $B_i$ and $G_i$, 
respectively, arranged in an alternating sequence.
This yields the ansatz 
\begin{align}  \label{eq:ansatz_class_unitary}
    U_{A_l}=U_{B_{l+1}}U_{G_l}U_{B_l}\cdots U_{G_1}U_{B_1}
\end{align}
where the length $l$ of the ansatz corresponds to the total number of CZLs.
By virtue of the group homomorphism property of \cref{eq:homomorphism} in the methods section, 
the symplectic representative of the ansatz circuit,
$U_{A_l}$, is simply given by
\begin{align}  \label{eq:ansatz_class_binary}
    {A_l}={B_{l+1}}{G_l}{B_l}\cdots {G_1}{B_1}\in\Sp(\FF_2^{2n}).
\end{align}

This concept is illustrated in \cref{fig:tt12:cx} for a use case example that will be discussed in detail below.
Here, we proceed by stating a straightforward observation.
\begin{lemma}[Expressivity of our ansatz class] \label{lem:ansatz}
    Consider a quantum device whose connectivity graph $\Gamma_\text{con}$ has just one connected component.
    Then, every $n$-qubit Clifford gate $U\in \Cn$ can be expressed as a hardware-tailored circuit, that is, there exist SCLs $\smash{B_1,\ldots, B_{l+1} \in \Sp(\FF_2^{2n})}$ and CZLs $\smash{G_1,\ldots,G_l \in \Sp(\FF_2^{2n})}$ with $\smash{\Gamma_1,\ldots, \Gamma_l \le \Gamma_\text{con}}$ such that $U= U_{B_{l+1}}U_{G_l} U_{B_l} \ldots U_{G_1}U_{B_1}$.
\end{lemma}
\begin{proof}
The $n$-qubit Clifford group is generated by single-qubit Clifford gates and CZ gates. 
Limited connectivity is mediated via SWAP gates (since $\Gamma_{\text{con}}$ is connected), and each SWAP decomposes into CZ and Hadamard gates, so all operations remain within our ansatz class.
\end{proof}

Although straightforward to prove, 
\cref{lem:ansatz} offers a simple guarantee that our ansatz class of hardware-tailored circuits $U_{A_l}$ from \cref{eq:ansatz_class_unitary} is expressive enough to implement all Clifford operations.
This motivates their use as templates in the search for logical Clifford gate implementations. 
While the proof of \cref{lem:ansatz} does not aim to minimize the circuit length $l$, we will see in \cref{fig:iceberg:bar_chart,fig:tt12:cx,fig:color832:gates} that small values of $l$ can often be achieved in practice.

\subsection{Characterization of target circuits}

Here, we scrutinize the operations that we aim to match to our class of ansatz circuits: logical Clifford gates.
Consider an  $\llbracket n,k,d\rrbracket$ stabilizer code with encoding operation $U_E$ as well as a $k$-qubit Clifford gate $U_C$ that we want to implement on the logical level.
A trivial (but never fault-tolerant) implementation is given by  
\begin{align} \label{eq:trivial_implementation}
{U_\text{triv} = U_E(U_C\otimes \mathbbm 1_{2^{n-k}})U_E^\dagger} \, ,    
\end{align}
i.e., by decoding the quantum information, applying $U_C$ on the unprotected qubits, and re-encoding.
As the symplectic representation is a group homomorphism,
the operator $U_\text{triv}$ is represented by $E C' E^{-1}$, where
\begin{align}
U_C\otimes \mathbbm 1_{2^{n-k}}\longmapsto C'
\end{align}
defines $C'\in \Sp(\FF_2^{2n})$.
Note that there might be other, more efficient circuit implementations of the logical gate that are also represented by $E C' E^{-1}$.
However, all of them have a fully-determined action not only on the code space $\mathcal{L}$ but on the entire physical $n$-qubit Hilbert space.
While the action on $\mathcal{L}$ is determined by the choice of the target gate $U_C$, the action on the ambient space is fixed by the choice of a particular \emph{gauge}.
Exploiting this gauge freedom is what will allow us to probe a vast amount of different implementations for $U_C$ on the logical level.
Indeed, given a second encoding matrix $\smash{  E_2\in\Sp(\FF_2^{2n})}$ for the considered code, we can implement the same logical gate via $\smash{E C'  E^{-1}_2}$.
This has the same effect as $E C'E^{-1} $ on the logical level,
but may act differently on the physical degrees of freedom.
The full characterization of this gauge freedom is our first main result:

\begin{restatable}[Characterization of target circuits via gauges]{theorem}{Fcomplete}\label{thm:gauge_freedom}
    Consider an $\llbracket n,k,d\rrbracket$ code with  a Clifford encoding circuit $U_E\in\Cn$ as well as a $k$-qubit Clifford gate $U_C\in\mathcal{C}^k$.
    Write $C'\in\Sp(\FF_2^{2n})$ for the matrix that represents $U_C\otimes \mathbbm1_{2^{n-k}}\in\Cn$.
    Then, every $n$-qubit Clifford gate that implements $U_C$ on the logical level is represented by $E C' F E^{-1}\in\Sp(\FF_2^{2n})$ for precisely one symplectic matrix of the form
    \begin{equation}
        \begin{tikzpicture}[baseline=(current  bounding  box.center)]
        \draw (0, 0) node {
        $F
        =\begin{bmatrix}
            \mathbbm{1} & \ast & \cdots & \ast & 0 & 0 & \cdots & 0 \\
            0 & \ast & \cdots & \ast & 0 & 0 & \cdots & 0 \\
            \vdots & \vdots &  \ddots& \vdots & \vdots & \vdots &  \ddots& \vdots\\
            0 & \ast & \cdots & \ast & 0 & 0 & \cdots & 0\\
            
            0 & \ast & \cdots & \ast & \mathbbm{1}
            
            & 0 & \cdots & 0\\
            \ast & \ast & \cdots & \ast & \ast & \ast & \cdots & \ast \\
            \vdots & \vdots & \ddots & \vdots & \vdots & \vdots &  \ddots & \vdots \\
            \ast & \ast & \cdots & \ast & \ast & \ast & \cdots & \ast \\
        \end{bmatrix}$
        };
        \draw [decorate,decoration={brace,amplitude=5pt,mirror}] (-0.94+1.82,-2.3) -- (0.29+1.82,-2.3) node [black,midway,yshift=-10px] {\footnotesize $n-k$};
        \draw [decorate,decoration={brace,amplitude=5pt,mirror}] (-1.37+1.82,-2.3) -- (-1.12+1.82,-2.3) node [black,midway,yshift=-10px] {\footnotesize $k$};
        \draw [decorate,decoration={brace,amplitude=5pt,mirror}] (-0.94,-2.3) -- (0.29,-2.3) node [black,midway,yshift=-10px] {\footnotesize $n-k$};
        \draw [decorate,decoration={brace,amplitude=5pt,mirror}] (-1.37,-2.3) -- (-1.12,-2.3) node [black,midway,yshift=-10px] {\footnotesize $k$};
        \draw [decorate,decoration={brace,amplitude=5pt,mirror}] (2.4,-2.1) -- (2.4,-0.05) node [black,midway, xshift=10px] {\footnotesize $n$};
        \draw [decorate,decoration={brace,amplitude=5pt,mirror}] (2.4,0.05) -- (2.4,2.1) node [black,midway, xshift=10px] {\footnotesize  $n$};
        \end{tikzpicture} ,
        \label{eq:freedom_matrix}
    \end{equation}
    where asterisk ($\ast$) symbols indicate binary variables that are only constrained by the requirement $F\in\Sp(\FF_2^{2n})$.
    Moreover, the set $\mathcal{F}=\{F \in \Sp(\FF_2^{2n}) \ \vert \ F \text{ obeys \cref{eq:freedom_matrix}}\}$ is a group, and there are exactly 
    \begin{align}
    \vert \mathcal{F}\vert = \frac{2^{n(n+1)/2 + k(n-k)}}{2^{k(k+1)/2}}\prod_{m=1}^{n-k}(2^{n-k}-2^{m-1})\label{eq:freedom_number}
    \end{align}
    valid choices for $F\in \mathcal{F}$.
\end{restatable}
\begin{proof}
    See Sec.~II in the supplementary material (SM).
\end{proof}

The matrix $F$ introduced in \Cref{thm:gauge_freedom} parameterizes the available gauge freedom of implementing a desired Clifford gate on the logical level whose action outside the code space is irrelevant.
For this reason, we refer to $F$ as \emph{freedom matrix}.
By \cref{eq:freedom_number}, there exist exponentially many valid assignments for the $4n^2$ entries of $F$.
As such, \Cref{thm:gauge_freedom} establishes a powerful handle for probing vast amounts of potential circuit implementations at the same time.
From a conceptual standpoint, it is worth noting that the \emph{freedom gauge group} $\mathcal{F}$, together with the (Pauli) stabilizer group $\mathcal{S}$, parameterizes the group $\mathcal{G}= 
\{U \in \Cn \ \vert \ \forall \ket{\psi}\in \mathcal{L}: U\ket{{\psi}}= \ket{\psi} \} $ of Clifford stabilizer symmetries.
More precisely, it holds
\begin{align}
\mathcal{G} = \{ f_P(F)e^{\iu\varphi(F)}  U_E U_FU_E^\dagger S \ \vert \ F\in \mathcal{F}, S \in \mathcal{S}\}
\end{align}
for an appropriate choice of $f_P: \mathcal{F} \rightarrow  \Pn$ and $\varphi: \mathcal{F} \rightarrow \RR$
to ensure the correct Pauli frame and global phase, respectively.

\subsection{The binary optimization program}

Now that we have identified $EC'FE^{-1}$ as a parameterization of all symplectic matrices representing a given logical Clifford gate,
we are in the position to devise methods for constructing concrete circuit implementations.
Up to Pauli operators, every ansatz circuit $U_{A_l}$ from \cref{eq:ansatz_class_unitary} solves this problem if
\begin{equation}
    A_l=EC'FE^{-1} \label{eq:main_basic}
\end{equation}
for a valid assignment of the freedom matrix $F$ from \cref{eq:freedom_matrix}.
In practice, solving \cref{eq:main_basic} for a suitable ansatz length $l$ amounts to finding assignments for the binary variables in the matrices $B_1,\dots,B_{l+1}$ and the adjacency matrices $\Gamma_1,\dots,\Gamma_{l}\leq\Gamma_\text{con}$.
The specific gauge fixed by $F$ is important insofar as it enables compatibility with many different physical circuits simultaneously.
It is not necessary to explicitly enforce $F\in\Sp(\FF_2^{2n})$ as this is always fulfilled if $A_l\in\Sp(\FF_2^{2n})$.
The latter is readily taken care of by adding $\smash{b_i^{(j)xx}b_i^{(j)zz}+b_i^{(j)xz}b_i^{(j)zx}=1}$ for every qubit $j$ and all SCLs $B_i$ as well as $\Gamma_i=\Gamma_i^T$ for all CZLs $G_i$ to the binary system of polynomial equations~\cite{miller_hardware_tailored_2024}.

Having identified \cref{eq:main_basic} as the mathematics behind the circuit construction problem is one of the core results of this work.
Now, we could formulate a binary optimization program for solving it.
Before we do so, however, let us first transform the system of equations in order to ease the problem for numerical solvers.
We begin by treating the product $C'F$ as a matrix in its own right, before we  remove columns and rows indexed from $k+1$ to $n$. 
This yields the matrix
\begin{align}  
    F'_C = \begin{bmatrix}
        C & 0 \\
        * & *
    \end{bmatrix} \in \mathbbm F_2^{(n+k)\times(n+k)} \label{eq:freedom_matrix_reduced}
\end{align}
with $(n+k)(n-k)$ parameterized entries in the lower block.
Next, we trim columns $k+1$ to $n$ from $E$ in \cref{eq:encoding_matrix_details}, which results in the 
matrix $E'\in \FF_2^{2n \times (n+k)}$ with
\begin{align}
    E' = \left[\begin{array}{c|c|c}
    \mathbf{x}_{1}\dots\mathbf{x}_{k} &\mathbf{z}_{1}\dots\mathbf{z}_{k}&\mathbf{s}_{1}\dots\mathbf{s}_{n-k}\\\hline
    \mathbf{x}'_{1}\dots\mathbf{x}'_{k}&\mathbf{z}'_{1}\dots\mathbf{z}'_{k}&\mathbf{s}'_{1}\dots\mathbf{s}'_{n-k}
    \end{array}\right] 
    .
    \label{eq:encoding_matrix_truncated}
\end{align}

With this notation in place, we are ready to state our next main result:

\begin{corollary}[Simplified system of equations] 
\label{crl:optimization}
    For a given encoding circuit $E$, desired logical gate $C$, and ansatz length~$l$, the polynomial system of equations over $\FF_2$ in \cref{eq:main_basic}
    defines the same variety of solutions for $A_l\in\Sp(\FF_2^{2n})$ as
    \begin{align}
    A_l E'=E'F'_C. \label{eq:main_improved}
    \end{align}
\end{corollary}
\begin{proof}
By construction, every solution $A_l$ of \cref{eq:main_basic} solves \cref{eq:main_improved} via \cref{eq:freedom_matrix_reduced}.
Conversely, let $A_l$ be a solution of \cref{eq:main_improved}.
Then it is readily verified that $U_{A_l}$ permutes the stabilizer group and that it transforms logical Pauli operators in the same way as $U_C$.
By \cref{app:lem:logical_clifford_gates,app:lem:different_logical_clifford_gates}
from the methods section, $U_{A_l}$ thus implements a $U_C$ gate on the logical level.
Hence, \Cref{thm:gauge_freedom} applies, which finishes the proof.
\end{proof}

With \Cref{crl:optimization} at hand, 
we can formulate the problem of finding \emph{hardware-tailored} (HT) logical Clifford gate implementations in terms of an IQCP.
After specifying a suitable (polynomial) cost function for the ansatz $A_l$ from \cref{eq:ansatz_class_binary}, this IQCP can be written as
\begin{align}
\begin{split}
    \text{min }\,\,& \operatorname{cost}(A_l) \\
    \text{subject to }\,\,& A_lE'=E'F_C',\\
    &B_1,\dots,B_{l+1} \text{ are SCLs,}\\
    \text{and }&G_1,\dots,G_l \text{ are HT CZLs}.
\end{split} \label{eq:binary_optimization_problem}
\end{align}

For the examples constructed in this work, we define $\operatorname{cost}(A_l)$ as the total number of CZ gates in $U_{A_l}$; however, other objectives are also conceivable, e.g., penalizing low-fidelity CZ or Hadamard gates, see methods section ``Cost function alternatives for the optimizer''.

The (binary) variables, which we optimize over in the IQCP, are given by the parameterized entries of $B_i$, $G_i$, and $F_C'$.
Behind the scenes of \cref{eq:binary_optimization_problem}, binary slack variables are introduced to reduce the degree of the polynomials in \cref{eq:main_improved} to quadratic.
Similarly, integer slack variables must be introduced to account for the fact that \cref{eq:main_improved} must be satisfied modulo 2~\cite{Optimization}.
Although solving IQCPs is NP-hard in the worst case, there exist sophisticated methods that can tackle it effectively in practice~\cite{ku_hybrid_exact_2017}.
In this work, we leverage a state-of-the-art IQCP solver provided by Gurobi~\cite{gurobi_gurobi_2023},
which significantly enhances the quality of the circuits we can construct, as demonstrated in \cref{fig:iceberg:bar_chart,fig:tt12:cx,fig:color832:gates}.
For typical runtimes of the IQCP solver, see SM Sec.~I.

Let us summarize the results established thus far.
We proposed an ansatz class of hardware-tailored circuits,
and characterized the gauge freedom of implementing a logical Clifford gate.
We brought these two concepts together by formulating an IQCP that can be solved in order to obtain hardware-tailored circuit implementations of a logical Clifford gate.
This circuit can be optimized with respect to, e.g., two-qubit gate count.
The input parameters of our circuit construction framework are a reduced encoding operation $E'$, a target logical gate $C$, the length $l$ of the ansatz circuit $A_l$, and the hardware connectivity $\Gamma_\text{con}$ of a quantum device.
The output of the IQCP is a circuit $A_l$ (implementing $C$ on the logical level)
given as an alternating sequence of SCLs $B_i$ and hardware-tailored CZLs $G_i$,
as well as a solution for the reduced freedom matrix $F'_C$.
The latter just fixes the gauge outside the code space and is usually not of practical interest.
Nevertheless, $F'_C$ can be inspected to analyze how $A_l$ permutes the stabilizer group.
A Gurobi-based implementation of this framework is available as a Python package and can be installed using \inlinecode{pip install htlogicalgates}.

\subsection{First use case: $\llbracket 4,2,2\rrbracket$ Iceberg code}

\begin{table}
    \centering
    \begin{tabular}{|c||cc|cc|}\hline
        \multirow{2}{*}{Connectivity} & \multicolumn{2}{c|}{CZ count} & \multicolumn{2}{c|}{Runtime}\\
                    &max&avg.& max                     & avg.        \\\hline\hline
         Star       & 6 &2.5 &$3600\unit{s}$           & \ $61\unit{s}$\ \\\hline
         Circular   & 4 &3.0 &$\phantom{1}436\unit{s}$ & \ $85\unit{s}$\ \\\hline
         Linear     & 5 &3.0 &$\phantom{1}508\unit{s}$ & \ $29\unit{s}$\ \\\hline
    \end{tabular}
    \caption{\justifying{}
    \textbf{Results for the $4$-qubit Iceberg code.}
    Circuit cost (CZ count) and classical preprocessing cost (runtime) for constructing hardware-tailored circuit implementations in the worst case (max) and on average (avg.) for all 720 logical Clifford gates of the  $\llbracket4,2,2\rrbracket$ Iceberg code for three different connectivities, see \cref{fig:iceberg:bar_chart}.
    Our circuit implementations are constructed by solving the IQCP in \cref{eq:binary_optimization_problem} using our Gurobi-based open-source software, applied to an ansatz circuit $A_l$ of length $l=3$ with a timeout of $3600\unit{s}$.
    All computations were carried out on four cores of an Intel Xeon CPU E5-2695 v2 @ $2.40\unit{GHz}$ with $20\unit{GB}$ of RAM.
    The solver performs reliably and fast.}
    \label{tab:iceberg:sim}
\end{table}

\begin{figure*}
    \centering
    \includegraphics{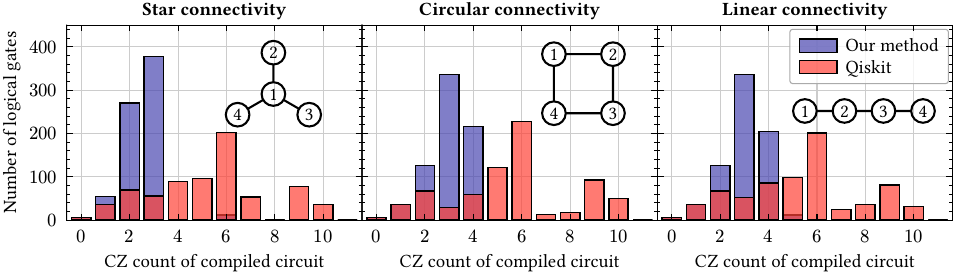}
    \caption{\justifying{}
    \textbf{Logic in the Iceberg code.}
    Distributions of CZ counts for all 720 logical Clifford gates (modulo Paulis and phases) of the $\llbracket 4,2,2\rrbracket$ Iceberg code.
    Note that the number of two-qubit gates differs from circuit to circuit, which leads to the displayed distributions.
    We compare our hardware-tailored circuit implementations from \cref{tab:iceberg:sim} (blue) 
    with those obtained by an off-the-shelf Qiskit method (red).
    The three hardware connectivities, to which the circuit implementations have been tailored, are shown as insets.
    Our method achieves lower two-qubit gate counts by incorporating the gauge freedom from \Cref{thm:gauge_freedom} into the optimization process.} 
    \label{fig:iceberg:bar_chart}
\end{figure*}

The first QECC for which we construct hardware-tailored logical Clifford gates is the four-qubit \emph{iceberg code}~\cite{rains_quantum_codes_1997}.
This $\llbracket4,2,2\rrbracket$ code belongs to a family of $\llbracket n,n-2,2\rrbracket$ codes with stabilizer generators $X^{\otimes n}$ and $Z^{\otimes n}$, where $n$ is even.
The $\llbracket4,2,2\rrbracket$ iceberg code has $k=2$ logical qubits and therefore $\vert \Sp(\FF_2^4)\vert = 720$ logical Clifford gates.
Each of them is implementable in $12,288$ different gauges, recall \Cref{thm:gauge_freedom}.
Leveraging our new techniques, we optimize over all gauges and a variety of circuit templates (ansätze) to identify circuit implementations that minimize the number of CZ gates.
To demonstrate the flexibility of our method, we consider three connectivities: \textit{star}, \textit{circular}, and \textit{linear}, as shown in the insets of \cref{fig:iceberg:bar_chart}.
For all three connectivities and every logical Clifford gate, we succeed in constructing a circuit implementation with no more than three CZLs and four SCLs, i.e., with an ansatz $U_{A_l}$ of length $l=3$.
In \Cref{tab:iceberg:sim}, we present the maximum and average two-qubit gate counts of the constructed circuits, along with the maximum and average runtime of the solver that found them. 
In all cases, we see that no more than six physical CZ gates are required for implementing a logical two-qubit Clifford circuit.
We observe no significant difference in the quality of the obtained circuits, which we attribute to the small number of physical qubits.

Regarding the runtime of our classical circuit constructor, it is important to note that the leveraged Gurobi solver operates in two phases.
First, it identifies a feasible solution to \cref{eq:binary_optimization_problem}, corresponding to a valid circuit implementation of the target logical gate. 
Then, it attempts to prove optimality by searching for better feasible points and, if successful, replaces the initial solution with an improved one.
Since we aim to construct $3 \times 720$ circuit implementations, we impose a one-hour timeout on the Gurobi solver.
As shown in \cref{tab:iceberg:sim}, this timeout is only reached in the case of star connectivity. 
Even then, it affects only the proof of optimality;
the solver still produced valid and high-quality solutions for all 720 logical Clifford gates.

We also compare our circuit implementations to a readily available baseline alternative.
For every logical Clifford gate, we use Qiskit~\cite{qiskit_qiskit_2023} 
to compress the trivial implementation $U_\text{triv}$ defined in \cref{eq:trivial_implementation}. 
Since Qiskit does not support optimization over gauges $F\in \mathcal{F}$, we fix the gauge to $F = \mathbbm{1}$.
This is the gauge in which every stabilizer is mapped to itself and every logical Pauli operator is transformed directly into the corresponding target logical operator, without additional stabilizer byproducts.
We do not attempt a brute-force search over all $12,288$ possible gauges because, to the best of our knowledge, an algorithm for exploring the full gauge freedom is introduced here for the first time.
After obtaining a circuit, we transpile it to the three hardware connectivities under consideration. 
Our source code of this approach is available on GitHub (see Ref.~\cite{htlogicalgates_github} for our Python code for compiling the Qiskit circuits from \cref{fig:iceberg:bar_chart}).
For each connectivity, we present two histograms in \cref{fig:iceberg:bar_chart}, showing the two-qubit gate counts for our method (blue) and the Qiskit baseline (red). 
Our circuits consistently achieve lower CZ counts than the off-the-shelf Qiskit functionalities, 
highlighting the advantage of our gauge-aware circuit compilation and optimization.
Furthermore, our circuits exhibit virtually no outliers (apart from just twelve instances with six CZ gates, two of which are due to a Gurobi timeout),
further underscoring the advantage of a global optimization approach over conventional circuit optimization techniques.

\subsection{Second use case: $\llbracket 12,2,3\rrbracket$ twisted toric code}

\begin{figure}
    \centering
    \includegraphics[scale=\typstscale]{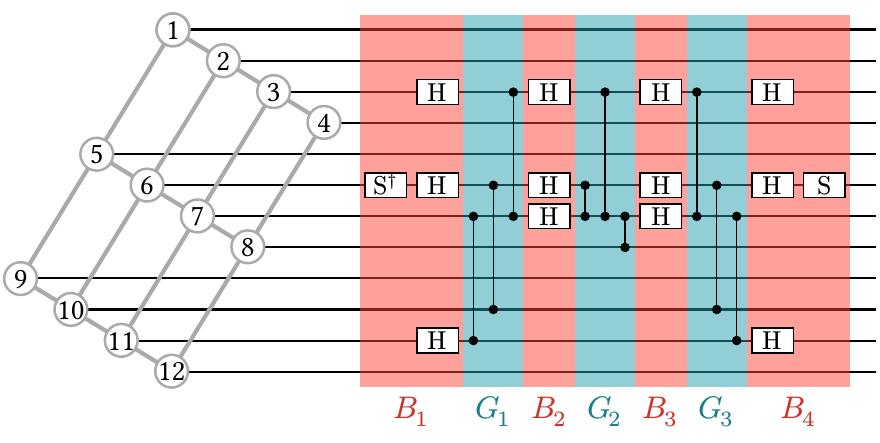}
    \caption{
    \justifying{} 
    \textbf{Logic in the twisted toric code.}
    A hardware-tailored circuit implementation of the $\overline{CX_{2,1}}$ gate for the $\llbracket12,2,3\rrbracket$ twisted toric code.
    The circuit (right) is tailored to a square-grid connectivity (left) and requires nine CZ gates.
    It was constructed by solving \cref{eq:binary_optimization_problem} using an ansatz $U_{A_l}$ with $l=3$ controlled-$Z$ gate layers (CZLs) and four single-qubit Clifford gate layers (SCLs).
    The evident symmetry suggests that computer-generated circuits like this might be generalizable to larger twisted toric codes.
    }
    \label{fig:tt12:cx}
\end{figure}

Next, we consider the $\llbracket12,2,3\rrbracket$ twisted toric code~\cite{breukmann_balanced_product_2021} 
and tackle the aforementioned addressability problem.
When constructing hardware-tailored logical circuit implementations for this code,
we do not explicitly exploit any of its symmetries.
Instead, we simply inform our solver for \cref{eq:binary_optimization_problem}  that the  stabilizer group is generated by
$X_1X_{2}X_6X_7$,
$X_1X_4X_{11}X_{12}$,
$X_2X_3X_9X_{10}$, 
$X_3X_{4}X_5X_8$,
$X_5X_{6}X_{10}X_{11}$, 
$Z_1Z_2Z_{9}Z_{12}$,
$Z_1Z_4Z_5Z_6$,
$Z_2Z_3Z_7Z_8$,
$Z_3Z_4Z_{10}Z_{11}$,
and $Z_{5}Z_8Z_9Z_{10}$,
and that the logical Pauli operators are chosen as 
$\overline{X}_1=X_1X_5X_9$, 
$\overline{Z}_1=Z_1Z_2Z_3Z_4$,
$\overline{X}_2=X_1X_2X_3X_4$, 
and 
$\overline{Z}_2=Z_2Z_6Z_{10}$.
From \cref{eq:freedom_number}, 
we know that for each logical Clifford gate, there exist approximately $1.5 \times 10^{58}$ different implementations that differ only in their action on states outside the code space.

Assuming a $3\times4$ square-grid connectivity,
we tailor circuit implementations of the logical controlled-$X$ gate with control qubit 2 and target qubit 1.
Note that our method is not limited to this example.
First, we consider an ansatz length $l=2$ and succeed in constructing a circuit with eleven CZ gates (not shown).
By increasing to $l=3$, we find an even shorter circuit with only nine CZ gates that is displayed in \cref{fig:tt12:cx}.
Interestingly, this computer-generated circuit appears to exhibit a nontrivial structure:
the first and the last SCLs are inverses of each other, i.e., $B_1=B_4^{-1}$.
The same is true for the (self-inverse) inner SCLs and the outer CZLs, i.e.,  $B_2=B_3^{-1}=B_3$ and  $G_1 =G_3^{-1}=G_3$.
This emergent structure spurs hope that, despite the NP-hardness of solving \cref{eq:binary_optimization_problem}, our software can be used to construct and analyze small-scale logical gates, and that these constructions, once understood, may be analytically generalized to larger codes in further work.
In this context, it is important that one can efficiently verify whether a candidate circuit implements a desired logical Clifford gate,
see \cref{app:lem:different_logical_clifford_gates} in 
methods section ``Characterization of logical Clifford gates''.
Finally note that our methods have already been applied for the construction of other logical gates for the twisted toric code in Ref.~\cite{andersen_small_quantum_low_2026} by other authors.

\subsection{Third use case: $\llbracket 8,3,2\rrbracket$ color code} 

\begin{figure*}
    \centering
    \includegraphics[width=\linewidth]{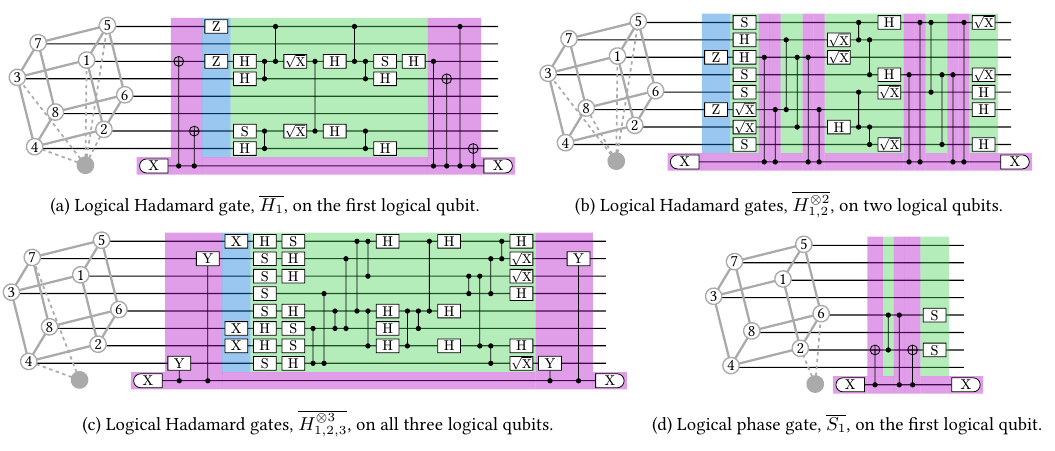}
    \caption{\justifying{}
    \textbf{Logic in the smallest interesting color code.}
    Hardware-tailored circuit implementation of weakly FT logical \textbf{a-c} Hadamard and \textbf{d} phase gates for the $\llbracket 8,3,2\rrbracket$ color code.
    By rotating the cube, other logical qubits can be addressed.
    The unitary Clifford subcircuits (green) are tailored to a cube connectivity (left)
    by solving the IQCP in \cref{eq:binary_optimization_problem}.
    Then, the Pauli frame (blue) is adjusted by applying Theorem~2 in Ref.~\cite{deheane_clifford_group_2003}.
    Finally, a suitable flag gadget (purple) is constructed by applying \cref{lem:ft_flag} to make the circuit implementation weakly fault-tolerant,
    where dashed lines in the graph on the left indicate the connectivity required by the flag gadget.
    Using Stim~\cite{gidney_stim_2021}, we verify that every  single fault in the circuit-level Pauli noise model (see methods section ``Weak fault tolerance with flag gadgets'') is either detectable or has a trivial logical action (see Ref.~\cite{htlogicalgates_github} for a circuit distance calculation).
    For the Hadamard gates, we additionally perform a circuit-level noise simulation as a sanity check of weak fault tolerance, see \cref{fig:color832:sim}.
    }
    \label{fig:color832:gates}
\end{figure*}

The final code considered in this paper is the $\llbracket 8,3,2\rrbracket$ color code, often referred to as the ``smallest interesting color code'' due to its remarkable ability of supporting a transversal non-Clifford gate~\cite{campbell_the_smallest_interesting_2016}.
More precisely, applying the operator $T\otimes T^\dagger\otimes T^\dagger\otimes T\otimes T^\dagger\otimes T\otimes T\otimes T^\dagger$ implements the gate $\mathrm{CCZ} = \mathrm{diag}(1,\ldots,1,-1)$ on the three logical qubits.
It is well known that the gate set comprising CCZ and Hadamard gates is universal in a certain sense~\cite{shi_both_toffoli_2003}; however, it is also worth noting that the group they generate contains only real matrices.
As such, there is value in augmenting the gate set with the Clifford gate $S = \mathrm{diag}(1, \iu)$.
Since the CCZ gate is already transversal, 
the Eastin–Knill theorem implies that the logical Hadamard gate for the $\llbracket 8,3,2\rrbracket$ color code must require a more complex circuit implementation~\cite{eastin_restrictions_on_transversal_2009}.

One way of implementing weakly FT logical Hadamard gates for the $\llbracket8,3,2\rrbracket$ color code relies on a projection-based (PB) teleportation approach~\cite{wang_fault_tolerant_2024}.
In this protocol, either one or two logical qubits are teleported into the $\llbracket 4,2,2 \rrbracket$ iceberg code, which supports a SWAP transversal two-qubit Hadamard gate. 
After the operation is applied, the qubits are teleported back into the color code.
We refer to these protocols as $\smash{\textit{Tele-}\overline{H}}$
and $\smash{\textit{Tele-}\overline{H^{\otimes 2}}}$ 
and provide their resource costs in \cref{tab:color832:counts}.
For example, the $\textit{Tele-}\overline{H}$ protocol has a CZ count of 26 and consumes a total of ten auxiliary qubits (four for the iceberg code and six for flagging).
Similarly, one can implement Hadamard gates on all three logical qubits of the $\llbracket8,3,2\rrbracket$ color code by applying 
first $\smash{\textit{Tele-}\overline{H^{\otimes 2}_{1,2}}}$ 
then  $\smash{\textit{Tele-}\overline{H_3}}$,
with costs (CZ count and consumed qubits) that simply add up.
If resets are available and parallelization is sacrificed,
only six auxiliary qubits are required at the same time.
More recently, the authors of Ref.~\cite{Butt_Measurement_Free_2026} constructed and experimentally implemented a measurement-free (MF) variant of the projection-based logical Hadamard gate.
Their circuit requires the same number of two-qubit gates (26) as the PB protocol; however, it reduces the number of auxiliary qubits from 10 to 4, see \cref{tab:color832:counts}.
Moreover, the MF circuit eliminates the need of measurements and classical feedback, which is advantageous in certain parameter regimes~\cite{butt_fault_tolerant_code_switching_2024, Butt_Measurement_Free_2026}.

\begin{table}[]
    \centering
    \begin{tabular}{|c||c|c|c|} \hline
       Gate  &  Hardware-tailored &  \multicolumn{2}{c|}{Teleportation-based}\\  
       \hline\hline
       & & \makebox[1.5cm][c]{PB~\cite{wang_fault_tolerant_2024}} & \makebox[1.5cm][c]{MF~\cite{Butt_Measurement_Free_2026}}
       \\ \hline
        $\overline{{H}}$ &  
        \begin{tabular}{lr}
            CZ count: & 13 \\
             Consumed ~qubits: & 1 \\
        \end{tabular} &
        \begin{tabular}{c}
             26 \\
             10 \\
        \end{tabular} & 
        \begin{tabular}{c}
             26 \\
             4 \\
        \end{tabular}
        \\ \hline
        $\overline{{H}^{\otimes 2}}$  &
        \begin{tabular}{lr}
            CZ count: &  16 \\
             Consumed qubits: & 1 \\
        \end{tabular} &
        \begin{tabular}{c}
             37 \\
             13 \\
        \end{tabular} &
        \begin{tabular}{c}
              \phantom{$^\ast$}52$^\ast$ \\
              \phantom{$^\ast$}8$^\ast$ \\
        \end{tabular}
        \\  \hline
        $\overline{{H}^{\otimes 3}}$ &
        \begin{tabular}{lr}
            CZ count: &  19 \\
             Consumed qubits: & 1 \\
        \end{tabular} &
        \begin{tabular}{c}
            \phantom{$^\ast$}63$^\ast$ \\
             \phantom{$^\ast$}23$^\ast$ \\
        \end{tabular} &
        \begin{tabular}{c}
             \phantom{$^\ast$}78$^\ast$ \\
             \phantom{$^\ast$}12$^\ast$ \\
        \end{tabular}
        \\ \hline
    \end{tabular}
    \caption{\justifying{}
    \textbf{Comparison of weakly FT logical Hadamard gates.}
     For the $\llbracket 8,3,2\rrbracket$ color code, we compare our hardware-tailored (HT) circuits against the
    \emph{projection-based} (PB) and \emph{measurement-free} (MF) circuits from Refs.~\cite{wang_fault_tolerant_2024, Butt_Measurement_Free_2026}.
    To our knowledge, so far only MF circuits have been demonstrated experimentally~\cite{Butt_Measurement_Free_2026}.
    We compare the total number of CZ gates in the weakly FT gadget as well as the number of consumed qubits.
    The latter refers to the number of single-qubit state initializations and, in the case of HT and PB circuits, also to the number of readouts.
    Whenever $\overline{H^{\otimes 2}}$ or $\overline{H^{\otimes 3}}$ is not available, 
    we report the resource cost of implementing the missing gate via a composition of $\overline{H}$ and/or $\overline{H^{\otimes 2}}$;
    these cases are highlighted by asterisk  
    ($\ast$) symbols.
    In contrast to previous approaches, our hardware-tailored circuit implementations require only a single auxiliary qubit to achieve weak fault tolerance, see \cref{fig:color832:gates}.
    }
    \label{tab:color832:counts}
\end{table}

With the methods developed in this paper, we are able to directly decompose logical Clifford circuits into physical ones, without relying on teleportation into a second code that supports these gates transversally.
In \cref{fig:color832:gates}, we present such teleportation-free implementations of single- and multi-qubit logical Hadamard gates on an arbitrary number of logical qubits
alongside a single-qubit logical phase gate.
Moreover, we present flag gadgets that make our circuits weakly FT in the sense {of \cref{def:fault_tolerance}}.
In all cases, a single flag qubit suffices to catch all undetectable errors 
that would be introduced by the hardware-tailored circuits alone.
In other words, the second flag qubit in the construction of \cref{lem:ft_flag} is not required in this context due to the absence of undetectable hook errors.
For the two-qubit logical Hadamard gate,
$\smash{\overline{H^{\otimes2}_{1,2}}}$,
shown in \cref{fig:color832:gates}, we need to apply \cref{lem:ft_flag} twice.
However, note that it is possible to reuse a single flag qubit without resetting it.
The total resource requirements of our teleportation-free circuits can be directly inferred from \cref{fig:color832:gates} and are provided in \cref{tab:color832:counts} for a direct comparison with the teleportation-based Hadamard gates from Refs.~\cite{wang_fault_tolerant_2024,Butt_Measurement_Free_2026}.
Our circuits consume significantly fewer qubits and require only about half as many physical CZ gates for the single- and two-qubit logical Hadamard gates.
To realize a three-qubit logical Hadamard gate using the teleportation-based approach, at least two circuits must be applied sequentially, resulting in additive resource costs.
This sequential approach can be avoided with the flexible method developed in this work, which yields a more than three times cheaper (in terms of CZ count) implementation of the three-qubit logical Hadamard gate.

To predict the performance of our circuit implementations, we carry out circuit-level simulations using  Stim~\cite{gidney_stim_2021}.
Our simulations are based on the error model described in SM Sec.~III,
where all physical error rates are proportional to a single parameter $p$.
We compare the following three methods for weakly FT mapping 
$\smash{\ket{\overline{0,0,0}}} $ to
$\smash{\ket{\overline{\mathord{+},\mathord{+},\mathord{+}}}}$, and vice versa,
and present the simulation results in \cref{fig:color832:sim}.
First, we apply the sequence of two teleportation-based circuits from Ref.~\cite{wang_fault_tolerant_2024} to implement 
 $\smash{\textit{Tele-}\overline{H_3}} \circ
 \smash{\textit{Tele-}\overline{H^{\otimes 2}_{1,2}}}$ (red stars).
Second, $\smash{\overline{H_3} \circ \overline{H^{\otimes 2}_{1,2}}}$
is implemented by sequentially applying the (rotated) circuits from \cref{fig:color832:gates} (yellow plus signs).
Finally, we also apply $\smash{\overline{H^{\otimes 3}_{1,2,3}}}$  in a single step (blue circles).
For details about weakly FT state preparation and readout, see SM Sec.~III.
In all cases, we observe in \cref{fig:color832:sim} the characteristic scaling of the logical error rate as $O(p^2)$.
This serves as a sanity check that every error arising from a single fault in the specified circuit-level noise model is either detected, so that the corresponding circuit execution is discarded in postprocessing, or logically trivial (\cref{def:fault_tolerance})---a property that we independently and rigorously confirm using Stim~\cite{gidney_stim_2021}.

The fraction of shots retained after discarding all circuit executions with violated detectors is plotted as a continuous curve in the background of \cref{fig:color832:sim}.
For both hardware-tailored options, we observe that this retained fraction decreases from nearly 100\% at $p=10^{-4}$ to about $10\%$ at $p=10^{-2}$.
The teleportation-based approach exhibits the same qualitative behavior, but with a significantly lower retained fraction throughout.
We attribute this to the higher number of noisy operations in these circuits.
Regarding the logical error rates, we observe that the hardware-tailored circuit performing all logical Hadamard gates simultaneously (blue circles) yields the best performance.
Strikingly, this represents an improvement of more than $ 4\times $ over the teleportation-based protocol.
A minor effect visible in \cref{fig:color832:sim} is that the error rates are slightly larger for the protocol mapping $\ket{\overline{0,0,0}}$ to $\ket{\overline{\mathord{+},\mathord{+},\mathord{+}}}$ (lower panel) than for the reverse direction (upper panel).
This suggests the presence of a larger number of detrimental error mechanisms when the circuits are applied to $\ket{\overline{0,0,0}}$.

\begin{figure}
    \centering
    \includegraphics{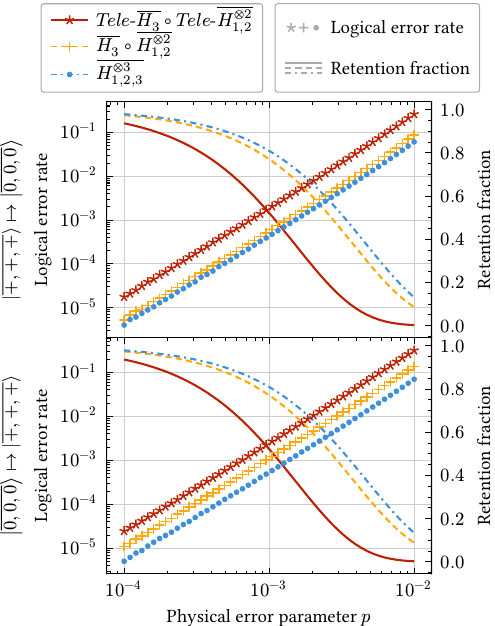}
    \caption{\justifying{} 
    \textbf{Hadamard gate simulations for the smallest interesting color code.}
    Logical error rate (left $y$-axis) and retained fraction (right $y$-axis) of shots after discarding executions with non-trivial syndromes, based on circuit-level simulations of the hardware-tailored gates presented in this work (yellow, blue) and the teleportation-based approach from Ref.~\cite{wang_fault_tolerant_2024} (red).
    The parameter $p$ ($x$-axis) and details of the simulated circuits are explained in SM Sec.~III.
    Each data point represents an average over $10^{8}$ circuit executions simulated using Stim~\cite{gidney_stim_2021}.
    All protocols are weakly fault-tolerant, and our teleportation-free approaches perform better due to their lower resource demands, see \cref{tab:color832:counts}.
    }
    \label{fig:color832:sim}
\end{figure}

\section{Discussion}

In this work, we developed powerful techniques to decompose logical Clifford circuits into physical ones for arbitrary stabilizer codes.
Starting from the symplectic representation of the Clifford group, we introduced a class of hardware-tailored ansatz circuits parameterized by binary variables. Similarly, we parameterized all possible gauges of a target logical gate by identifying the group of logical Clifford stabilizers associated with the given code.
This framework ultimately reduces circuit construction to solving and optimizing an IQCP.
We provide an open-source implementation, 
available as a Python package on GitHub~\cite{htlogicalgates_github}.

We have demonstrated the viability of our approach across a variety of gates and QECCs.
The selected use cases serve as proof-of-principle demonstrations, highlighting different strengths of our new techniques. 
First,  we constructed the full logical Clifford group for the $\llbracket4,2,2\rrbracket$ iceberg code under various connectivity constraints.
This illustrates the flexibility of our method and demonstrates that the selected circuit implementations are not cherry-picked.
Then, we presented a logical CX gate for the $\llbracket12,2,3\rrbracket$ twisted toric code. 
This shows that our methods scale to system sizes accessible to leading experimental academic groups working on superconducting qubits~\cite{Rosen_Synthetic_Magnetic_Vector_Potential_2024,
Besedin_lattice_surgery_2026,
Wang_demonstration_2026}
and that they
effectively tackle the so-called addressability problem:
how can one implement logical gates for QECCs whose logical qubits are delocalized across all physical qubits?
Finally, we constructed logical Hadamard gates for the $\llbracket8,3,2\rrbracket$ color code and made them weakly fault-tolerant by carefully applying \cref{lem:ft_flag} from methods section ``Weak fault tolerance with flag gadgets''. 
This represents a significant circuit engineering milestone for early-FT experiments with the $\llbracket8,3,2\rrbracket$ code, where highly-efficient weakly FT logical Hadamard gate implementations were previously lacking. 
Through circuit-level noise simulations, 
we have shown that our constructions not only consume significantly fewer auxiliary qubits than previous teleportation-based approaches but also considerably reduce the logical error rate.
While teleportation-based approaches can reduce error mechanisms that are not modeled by circuit-level simulations---e.g., 
coherent errors~\cite{chang_taming_coherent_noise_teleportation_2025}, leakage~\cite{lacroix_scaling_logic_color_2025}, and the heating and loss of ions~\cite{ryananderson_high_fidelity_2024} or atoms~\cite{Bluvstein_2025}---they require additional resources, most notably increased qubit overhead and a larger number of two-qubit gates in the circuit.
Whether these benefits justify the resource overhead is a question for future experiments.

From a broader perspective, the approach introduced here builds upon  ideas from global optimization to identify highly efficient, hardware-tailored circuits for the implementation of QECCs~\cite{miller_hardware_tailored_2024}. 
It complements circuit design methodologies based on algebraic rewrites—such as those using the ZX calculus~\cite{ZX, rodatz_floquetifying_stabiliser2024} or three-colored formalisms~\cite{Ruby}---which sequentially manipulate and optimize circuits through structured transformations.

As a result, it provides a flexible starting point for in-depth analyses of stabilizer codes and their logical Clifford gates under realistic hardware constraints. 
In future work, our approach may be adapted to address related problems in circuit discovery. 
For instance, by modifying the IQCP presented in this paper, it may be possible to construct hardware-tailored state preparation circuits, addressing the well-studied problem of fault-tolerant logical state preparation~\cite{zen_quantum_circuit_2024, peham_automated_synthesis_2025}. 
Moreover, our framework could potentially be extended to design hardware-tailored circuits for code switching~\cite{anderson_fault_tolerant_2014,
bombin_gauge_color_2015,
butt_fault_tolerant_code_switching_2024,
tan_single_shot_universality_quantum_ldpc_2025,
Heu_efficient_fault_tolerant_2025}.
Finally, it would be valuable to generalize some of the methods developed in this work to non-Clifford gates.

\section{Methods}

\subsection{Clifford group and stabilizer codes}

Here, we review some well-known mathematical concepts to introduce the necessary notation for the formulation of the circuit construction problem as an optimization program.
Let us start with the $n$-qubit Pauli group
\begin{align}
\Pn =\{\iu^q {X} ^\mathbf{r} {Z}^{\mathbf{r}'} \ \vert \ q \in \{0,1,2,3\},\ \, \mathbf{r},{\mathbf{r}'}\in \FF_2^n \},
\end{align}
where $\FF_2 $ is the binary field,
${X}^\mathbf{r}={X}^{r_1}\otimes \ldots\otimes {X}^{r_n}$
denotes an ${X}$-type $n$-qubit Pauli operator, and similarly for Pauli-$Z$. 
We define the \emph{binary representation 
of the Pauli group} as 
\begin{align}
    \Pn \longrightarrow \FF_2^{2n}, 
    \hspace{8mm} \iu^q{X}^\mathbf{r}{Z}^{\mathbf{r}'} \longmapsto  
    \begin{bmatrix}
        \mathbf{r} \\
        {\mathbf{r}'}
    \end{bmatrix}.\label{eq:pauli_representation_definition}
\end{align}

The $n$-qubit Clifford group, $\mathcal{C}^n$, is defined as the normalizer of the Pauli  group.
Modulo global phases and Pauli operators, the elements of $\Cn$ are in one-to-one correspondence with the binary symplectic group 
\begin{align}
    \Sp(\mathbbm F_2^{2n}) &= \left\{A\in
    \FF_2^{2n\times 2n}
    \ \big\vert \
    A^T \left[\begin{smallmatrix}
        0&\mathbbm 1\\
        \mathbbm 1&0
    \end{smallmatrix}\right] 
    A =
    \left[\begin{smallmatrix}
        0&\mathbbm 1\\
        \mathbbm 1&0
    \end{smallmatrix}\right] 
    \right\}.
\end{align}

For better readability, we break down the symplectic matrix $A\in\Sp(\FF_2^{2n})$ which represents a certain Clifford unitary $U\in\Cn$ into four blocks
\begin{align}
A=\begin{bmatrix}
        A^{xx}&A^{xz}\\
        A^{zx}&A^{zz}
    \end{bmatrix} \in \Sp(\FF_2^{2n})
\end{align}
with $A^{xx} ,  A^{xz}, A^{zx}, A^{zz} \in\FF_2^{n\times n}$.
The elements of the symplectic matrix $A$ are defined by requiring that
\begin{align}
    U{X}^\mathbf{r}{Z}^{\mathbf{r}'}U^\dagger \propto {X}^{A^{xx}\mathbf{r}+A^{xz}{\mathbf{r}'}}{Z}^{A^{zx}\mathbf{r}+A^{zz}{\mathbf{r}'}}\
    \label{eq:clifford_representation_definition}
\end{align}
holds for all Pauli operators represented by $\mathbf{r},{\mathbf{r}'}\in\mathbbm F_2^n$.
This defines the \emph{symplectic representation of the Clifford group},
\begin{align}
    \Cn \longrightarrow  \Sp(\FF_2^{2n}),\hspace{8mm}U \longmapsto A. \label{eq:homomorphism}
\end{align}
We use the suggestive notation $U=U_A$ whenever a Clifford operator $U$ is mapped to $A$.
Importantly, \cref{eq:homomorphism} is a {group homomorphism}, that is, $U_A U_B \mapsto AB$ holds for all symplectic matrices $A,B \in \Sp (\FF_2^{2n})$.
Note that the representation $U_A\mapsto A$ is not faithful; its kernel consists of global phases together with the Pauli group~\cite{calderbank_quantum_error_1997, deheane_clifford_group_2003}.
However, we can safely ignore the Pauli gates not explicitly handled, as they can be efficiently reconstructed when needed by applying Theorem~2 in Ref.~\cite{deheane_clifford_group_2003}.

Next, we need to briefly review the stabilizer formalism~\cite{gottesman_stabilizer_codes_1997}.
An $\llbracket n,k,d\rrbracket$ stabilizer code  is a $k$-qubit subspace $\mathcal{L} \subset (\mathbb C^{2})^{\otimes n}$ that is defined as the common $+1$-eigenspace of   $n-k$ commuting, independent, and Hermitian $n$-qubit Pauli operators $S_1,\ldots, S_{n-k}$.
The latter are called the  stabilizer generators of the code and 
they generate its stabilizer group $\StabGroup=\langle S_1, \ldots, S_{n-k} \rangle$. 
Finally, the parameter $d$ refers to the distance of an $\llbracket n,k,d \rrbracket$ code and is defined as the smallest number of qubits that need to be altered to cause a logical error.
The logical Pauli group $\langle \iu \mathbbm 1 , \overline{X}_j,  \overline{Z}_j\mid j \in \{1,\ldots, k\} \rangle$ is defined as the normalizer of the stabilizer group in the Pauli group, followed by modding out $\StabGroup$.
The choice of $\overline{X}_j$ and $\overline{Z}_j$ defines the computational basis of the logical qubits~\cite{calderbank_quantum_error_1997}.

It is a well-known fact that an $n$-qubit Clifford operator $U$ implements a logical operation on $\mathcal{L}$ iff $U$ permutes the stabilizer group~\cite{lidar_brun_quantum_error_2013}.
In this situation, the action of $U$ on the subspace $\mathcal{L}$ is fully determined by how it transforms the logical Pauli operators, i.e., by $U\overline{X}_j U^\dagger$ and $U\overline{Z}_j U^\dagger$ for all $j\in\{1,\ldots, k\}$.
The converse statement, however, is only true modulo stabilizer operators, see \cref{app:lem:different_logical_clifford_gates} below.

Any operation $U_E$ which maps $k$ qubits in a state vector $\ket{\psi}$ and $n-k$ auxiliary qubits in $\smash{\ket{0}^{\otimes(n-k)}}$ to the corresponding logical state vector $\ket{\overline{\psi}}\in\mathcal{L}$ is called an encoding operation for the considered stabilizer code.
It turns out that all stabilizer codes admit Clifford encoding operations~\cite{gottesman_stabilizer_codes_1997} and, in this work, we restrict ourselves to such operations.
This justifies the notation $U_E$ as we can use the symplectic representation $U_E\mapsto E$ to obtain the symplectic matrix $E\in\operatorname{Sp}\left(\FF_2^{2n}\right)$ of the encoding operation $U_E$.
Let us take a closer look at the encoding matrix 
\begin{align}
    E = \left[\begin{array}{c|c|c|c}
    \mathbf{x}_{1}\dots\mathbf{x}_{k}& \ast &\mathbf{z}_{1}\dots\mathbf{z}_{k}&\mathbf{s}_{1}\dots\mathbf{s}_{n-k}\\\hline
    \mathbf{x}'_{1}\dots\mathbf{x}'_{k}& \ast &\mathbf{z}'_{1}\dots\mathbf{z}'_{k}&\mathbf{s}'_{1}\dots\mathbf{s}'_{n-k}
    \end{array}\right],
    \label{eq:encoding_matrix_details}
\end{align}
where the logical Pauli operators $\overline{X}_j$ and $\overline{Z}_j$ are represented by their binary vectors $\mathbf{x}_j$, $\mathbf{x}'_j$ and $\mathbf{z}_j$, $\mathbf{z}'_j$, respectively, and the stabilizer generators $S_j$ are represented by $\mathbf{s}_j$ and $\mathbf{s}'_j$.
The other columns are less important for us and are abbreviated by an asterisk symbol ($\ast$).
For every stabilizer code and choice of logical Pauli operators, there exist many valid possibilities for selecting an encoding matrix $E$.
In \Cref{thm:gauge_freedom} of the main text, we formalize this observation and fully parameterize all available gauges relevant to our purposes by introducing the new concept of a freedom matrix $F$.

\subsection{Characterization of logical Clifford gates} 

In this methods section, we review several well-known results on logical gates, with an emphasis on Clifford operations for stabilizer codes.
This provides the necessary background to understand the origin of their gauge freedom, which is fully characterized in \Cref{thm:gauge_freedom} and further simplified in \Cref{crl:optimization}.
Throughout this section, let $\mathcal{L}$ be the code space of an $\llbracket n,k,d\rrbracket$ stabilizer code,  $ \{S_1,\ldots, S_{n-k}\}$ a set of stabilizer generators, and $\mathcal{S}=\langle S_1,\ldots, S_{n-k}\rangle $ its stabilizer group.
Moreover, let $\smash{\overline{X_1},\ldots,\overline{X_k}}$ and $\smash{\overline{Z_1},\ldots,\overline{Z_k}}$ denote a choice of logical Pauli-$X$ and -$Z$ operators for $\mathcal{L}$, respectively.

Let us start with the following lemma, which will only be used in this methods section.
\begin{lemma}[Inverses of logical gates are also logical gates] \label{app:lem:logical_gates_inverse}
    Let $U\in \mathcal{U}(2^n)$ be an $n$-qubit unitary.
    Then, $U$ is a logical gate for $\mathcal{L}$ if and only if (iff) the same is true for $U^\dagger$.
\end{lemma}
\begin{proof}
    It suffices to prove that the condition is sufficient as the roles of $U$ and $U^\dagger$ are interchangeable.
    Thus, let $U$ be a logical gate, i.e., for all code words $\smash{\ket{{\psi}}} \in \mathcal{L}$ it holds  $U\smash{\ket{{\psi}}} \in \mathcal{L}$.
    Let $\mathcal{B} = \smash{\{\ket{{\psi_1}} , \ldots, \ket{{\psi_{2^k}}}  \}}$ be a vector space basis of $\mathcal{L}$. 
    Then, $U\mathcal{B} = \smash{\{U\ket{{\psi_1}} , \ldots, U\ket{{\psi_{2^k}}} } \}$ is a subset of $\mathcal{L}$.  
    Since $U $ is injective, $U\mathcal{B}$ is linearly independent.
    Because of $\vert U\mathcal B\vert = 2^k$, it follows that $\mathcal{B}'=U\mathcal{B}$ is a basis of $\mathcal{L}$.
    By construction, the inverse of $U$  maps $\mathcal{B}'$ to $U^\dagger \mathcal{B}'=\mathcal{B} \subset \mathcal{L}$. 
    By linearity, it follows that $U^\dagger\ket{{\psi}}$ lies in $\mathcal{L}$ for all $\ket{{\psi}}\in \mathcal{L}$, which finishes the proof.
\end{proof}

Next, we formulate a useful condition for verifying that a Clifford circuit implements a gate on the logical level.

\begin{lemma}[Logical Cliffords permute the stabilizer group] \label{app:lem:logical_clifford_gates}
Let $U\in \Cn$ be an $n$-qubit Clifford gate. Then, the following conditions are equivalent:
    \begin{itemize}
        \item[(i)] $U$ is a logical gate for $\mathcal{L}$.
        \item[(ii)] $U$ is a logical Clifford gate for $\mathcal{L}$.
        \item[(iii)] For all $S\in \mathcal{S}$, we have $USU^\dagger\in \mathcal{S}$.
        \item[(iv)] For all $S\in \{S_1,\ldots, S_{n-k}\}$, we have $USU^\dagger\in \mathcal{S}$.
    \end{itemize}
\end{lemma}
\begin{proof} The implications 
    ``\textit{(ii)$\,\Rightarrow$(i)}'' and 
    ``\textit{(iii)$\,\Rightarrow$(iv)}'' are trivial.
    The reverse implications follow from the fact that $UPU^\dagger$ is a Pauli operator whenever $P$ is, and from a standard argument that expands an arbitrary stabilizer operator $S\in \mathcal{S}$ into a product of stabilizer generators. 
    Let us now prove the remaining implications.
    
    \textit{(i)$\,\Rightarrow$(iii):} 
    Let $S\in \mathcal{S}$.
    Since $U$ is a Clifford gate, $USU^\dagger$ is a Pauli operator. 
    Let us show that $USU^\dagger$ stabilizes the code space.
    Thus, let $\ket{\smash{\psi}}\in \mathcal{L}$ be an arbitrary code word. 
    From \cref{app:lem:logical_gates_inverse}, we know that $U^\dagger$ is a logical operator.
    Hence, we have $U^\dagger \smash{\ket{{\psi}}}\in \mathcal{L}$ and, therefore, 
    $SU^\dagger \smash{\ket{{\psi}}} =U^\dagger \smash{\ket{{\psi}}} $.
    This, in turn, implies  $USU^\dagger \smash{\ket{{\psi}}} = UU^\dagger {\ket{{\psi}}} = \ket{\smash{\psi}}$.
    In other words, $USU^\dagger$ is 
    a stabilizer operator.
    
    \textit{(iii)$\,\Rightarrow$(i):} 
    Let $\ket{\smash\psi}\in \mathcal{L}$.
    We have to show  $U\ket{\smash{\psi}}\in \mathcal{L}$.
    For every $S \in \mathcal{S}$, there exists some $S'\in \mathcal{S}$ such that $ SU = US'$ by assumption (iii).
    This implies
    $SU\ket{\smash{\psi}}  = U(U^\dagger SU)\ket{\psi}= 
    US'\ket{\psi} = U\ket{\psi}$.
    In other words, the state vector $U\ket{\psi}$ lies in the $+1$-eigenspace of all operators $S\in \mathcal{S}$, i.e., in the code space $\mathcal{L}$.
    This finishes the proof.
\end{proof}

Finally, we apply Schur's lemma~\cite{fulton_harris_representation_theory_2013} to prove that   Clifford gates act the same on the logical level iff they transform the logical Pauli operators identically, up to stabilizers.
\begin{lemma}[Equivalence of different logical 
Clifford
gates] \label{app:lem:different_logical_clifford_gates}
   Let $U,V\in \Cn$ be two logical Clifford gates for $\mathcal{L}$.
   The following conditions are equivalent:
    \begin{itemize}
        \item[(i)] There is a global phase $\alpha\in\RR$ such that for all $\ket{\psi}\in \mathcal{L}$ it holds  $U\ket{\psi} = e^{\iu \alpha} V\ket{\psi}$.
        
        \item[(ii)] For every logical qubit $i \in \{1,\ldots, k\}$, there exist stabilizer operators $S,S' \in \mathcal{S}$ with $U \overline{X_i}U^\dagger = V\overline{X_i}V^\dagger S $ and   $U \overline{Z_i}U^\dagger =  V\overline{Z_i}V^\dagger S'$. 
    \end{itemize}
\end{lemma}
\begin{proof}
    \textit{(i)$\,\Rightarrow$(ii):} 
    It suffices to show that  $P=U\overline{X_i}U^\dagger V\overline{X_i} V^\dagger$ is a stabilizer operator; the case of $\overline{Z_i}$ can be treated the same.
    Thus, let $\ket{\psi}\in \mathcal{L}$ be an arbitrary code word.
    By assumption, we have $P\ket{\psi}=(e^{\iu \alpha}V)\overline{X_i}(e^{\iu \alpha} V)^\dagger V\overline{X_i}V^\dagger \ket{\psi}$ because this calculation takes place in $\mathcal{L}$ entirely.
    Therefore, $P\ket{\psi}=\ket{\psi}$, which implies $P\in \mathcal{S}$,
    as claimed.
    
    \textit{(ii)$\,\Rightarrow$(i):} 
    We want to apply Schur's lemma to show that the linear map $f: \mathcal{L} \rightarrow \mathcal{L}, \ket{\psi}\mapsto V^\dagger U\ket{\psi}$ is proportional to $\id_\mathcal{L}$.
    Then, the proportionality constant must be of the form $e^{\iu \alpha}$ because $f$ is unitary.
    For this, we point out that the representation $g: G= \langle \overline{X_j},\overline{Z_j} \ \vert \ j\in \{1,\ldots, k\}\rangle \rightarrow \GL(\mathcal{L})$ that sends an $n$-qubit Pauli operator to itself is irreducible.
    We need to show that $f$ and $g$ commute.
    Thus, let $P\in G$ and $\ket{\psi}\in \mathcal{L}$ be arbitrary.
    By assumption, there is some $S\in \mathcal{S}$ with $UPU^\dagger = VPV^\dagger S$.
    This yields
    $ f(g(P)\ket{\psi}) = V^\dagger U P \ket{\psi} = V^\dagger (U P U^\dagger) U \ket{\psi} 
    = V^\dagger (VPV^\dagger S)U \ket{\psi} =  PV^\dagger U\ket{\psi}  = g(P)f(\ket\psi)$.
    Therefore, Schur's lemma applies, which finishes the proof.
\end{proof}

\subsection{Cost function alternatives for the optimizer}

\Cref{eq:binary_optimization_problem} in the main text formalizes the logical circuit compilation problem in terms of an IQCP.
As such, one can flexibly select a suitable cost function, as long as it is a polynomial function in the binary variables that parameterize the ansatz circuit.
For all examples constructed in this paper, we minimize the total number of CZ gates
because this is usually the noisiest component in current experimental platforms~\cite{ryananderson_high_fidelity_2024,bluvstein_logical_quantum_2024,google_quantum_error_2025}.
This number of CZ gates in an ansatz circuit $A_l$ [recall \cref{eq:ansatz_class_binary}] can be expressed as
\begin{align}
\operatorname{cost_{CZ}}\left(A_l\right)&=\frac{1}{2}\sum_{i,j=1}^n  \sum_{k=1}^{l}  \Gamma_{i,j}^{(k)} \, ,
\label{app:eq:cost_cz}
\end{align}
which is a sum in $\RR$, where  $\Gamma_{i,j}^{(k)}=1$ if the CZ gate between qubits $i$ and $j$ is present in the $k$-th CZL of $A_l$, and $\smash{\Gamma_{i,j}^{(k)}=0}$ otherwise.
We stress the fact that our method is not limited to the cost function in \cref{app:eq:cost_cz}.
For example, it can be straightforwardly modified to handle the case in which the CZ gate between different qubits $i$ and $j$ has a qubit-pair-dependent fidelity $f_{i,j} \in [0,1]$.
For simplicity, assume that the single-qubit gates in the circuit are perfect (the case of Hadamard can be treated analogously, see below).
Then, the fidelity of the circuit can be approximated as the product of the fidelities of the gates in the circuit, i.e.,  by $\prod_{i<j}\prod _{k} ({f_{i,j}})^{\Gamma^{(k)}_{i,j}}$.
Since the logarithm is strictly monotonic, this product is maximized iff
\begin{align} \label{app:eq:cost_cz_fidelity}
\operatorname{cost_{weighted-CZ}}(A_l ) = - \frac{1}{2} \sum_{i,j=1}^n   \log( {f_{i,j}}) \sum_{k=1}^l \Gamma_{i,j}^{(k)}
\end{align}
is minimized. 
Note that \cref{app:eq:cost_cz} and \cref{app:eq:cost_cz_fidelity} yield the same optimal ansatz circuit $A_l$ when all two-qubit gate fidelities are identical.

Finally, let us explain how to incorporate Hadamard gates into the cost function.
Here, it is important to note that we have 
$b^{xz}_{i,k} = 1$ iff a Hadamard gate on qubit $i$ is present in the $k$-th SCL of $A_l$.
Indeed, by Tab.~I of Ref.~\cite{miller_hardware_tailored_2024}, 
we have 
$b^{xz} = 1$ if the symplectic matrix $
\left[\begin{smallmatrix}
         b^{xx}&b^{xz}\\
         b^{zx} & b^{zz}
    \end{smallmatrix}\right] 
$
represents (modulo Paulis) a unitary in the set $\{H, SHS, SH,HS\}$. 
In the opposite case ($b^{xz}=0$) it represents a unitary in the set $\{I, S\}$.
This example demonstrates the flexibility of our approach.

\subsection{Weak fault tolerance with flag gadgets}

The flexible Clifford circuit synthesis method developed in this paper applies to stabilizer codes of arbitrary distance.
In the main text, we have used it to construct
circuits for codes of distance $d=2$ and $d=3$.
Within the scope of this work, however, we consider fault tolerance only for the $d=2$ case.
For this purpose, 
it suffices to augment our (initially) unitary circuits with custom \emph{flag gadgets}~\cite{chao_quantum_error_2018, chamberland_flag_fault_tolerant_2018, chao_flag_fault_tolerant_2020, anker_flag_gadets_2024, pato_optimization_tools_2024}  
to obtain \emph{weakly fault-tolerant} circuits~\cite{gerhard_weakly_fault_tolerant_2026}.
In this section, 
we review the definition of weak fault tolerance (\cref{def:fault_tolerance})
and show that it can be achieved for such circuits  (\cref{lem:ft_flag}).

Let us begin with the circuit-level noise model~\cite{lidar_brun_quantum_error_2013}.
Faults are modeled by the insertion of Pauli errors into the circuit, which is decomposed into individual gates.
Incoming qubits, single-qubit gates, and resets can each introduce one of three Pauli errors ($X$, $Y$, and $Z$), a measurement may produce a faulty outcome, and two-qubit gates can introduce one of fifteen non-trivial two-qubit Pauli errors ($IX, \ldots, ZZ$),
see SM Sec.~III for further details.
For a distance-2 code, one cannot in general guarantee detection of all residual errors arising from two or more faults.
For a single fault, however, one can still impose a meaningful fault-tolerance condition:

\begin{definition}[$d=2$, weak fault tolerance~\cite{gerhard_weakly_fault_tolerant_2026}] \label{def:fault_tolerance}
    Let $\mathcal{S}$ be the stabilizer group of an $\llbracket n,k,2\rrbracket$ code.
    Under the circuit-level Pauli noise model specified above, a circuit implementing the intended logical gate is called \emph{weakly fault-tolerant} if, for every encoded input state and every single allowed fault at any circuit location, with all other locations fault-free, at least one of the following conditions holds.
\begin{enumerate}[label=(\roman*)]
    \item \emph{The fault raises a flag.}
    In practice, we modify the original circuit produced by solving \cref{eq:binary_optimization_problem} by adding \emph{flag qubits} that interact with the data qubits through multi-qubit gates~\cite{chao_quantum_error_2018, chamberland_flag_fault_tolerant_2018, chao_flag_fault_tolerant_2020, anker_flag_gadets_2024, pato_optimization_tools_2024}.
    Each flag qubit is initialized in a known state and subsequently measured in a prescribed basis.
    A valid flag gadget must satisfy two requirements in the absence of faults: all flag measurement outcomes are deterministic for every encoded input state, and the modified circuit implements the same logical transformation as the original circuit.
    A flag is raised when at least one measurement outcome differs from its fault-free value, thereby heralding that a fault has occurred.
    In this case, we say that {``the error is flagged''}.

    \item \emph{The residual error anticommutes with at least one stabilizer.}
    The fault causes a single- or two-qubit Pauli error at a location somewhere within the circuits.
    Let $E$ denote the  effective Pauli error  that propagates out of the Clifford circuit~\cite{miller_propagation_of_2018}.
    We require the existence of a stabilizer $S\in\mathcal{S}$ such that $ES=-SE$, 
    to ensure that measuring the output stabilizers reveals the presence of $E$.
    The distinction from full fault tolerance is that we only require error propagation to preserve detectability (but not correctability) across successive gadgets.
    Therefore, $E$ is allowed to act on arbitrarily many data qubits as long as 
    this anticommutativity condition is satisfied~\cite{gerhard_weakly_fault_tolerant_2026}.
    In this case, we say that  {``the residual error is detectable''}.

    \item \emph{The residual error}, 
    $E$,
    \emph{is stabilizer-equivalent.}
    An undetectable error is permitted as long as it does not alter the encoded information.
    For Pauli errors and Clifford circuits, the effective error $E$ has a trivial logical action
    iff $E=e^{\iu \varphi}S$ for some $S\in\mathcal{S}$ and $\varphi \in \RR$.
    In this case, we say that {``the error has a trivial logical action''}.
\end{enumerate}
\end{definition}

Throughout this manuscript, we refer to an error or fault as \emph{``detectable''} if it satisfies condition (i), condition (ii), or both, in \cref{def:fault_tolerance}.
Note that the conditions (i)---(iii) 
allow error spreading to more and more physical qubits while excluding undetected non-trivial logical errors caused by a single fault.
Therefore, no single fault in the specified circuit-level fault model can remain unflagged while producing an undetectable, non-trivial logical error.

Next, we discuss how a given unitary circuit can be made weakly FT.
This question complements the main focus of the present work: constructing low-overhead circuits that implement a desired logical gate.
Having obtained such circuits, we now turn to how they can be equipped with suitable flag gadgets.
For this, we  take inspiration from Refs.~\cite{chao_fault_tolerant_quantum_2018, roffe_the_coherent_2019, debroy_extended_flag_2020, gonzales_quantum_error_2023, vandenberg_single_shot_2023, martiel_low_overhead_2025} to formulate the following result.
As before, consider an $\llbracket n,k, 2\rrbracket$ code and a Clifford circuit $U\in \Cn$ that implements a logical gate.
Our goal is to make $U$ weakly FT.
By \cref{def:fault_tolerance}, we need to ensure that every single error either (i) triggers a flag, (ii) has a detectable residual error, or (iii) has a trivial logical action.
Assume that the circuit is not yet weakly FT.
Then, there exists some problematic $n$-qubit Pauli error, $E\in \Pn$, that arises from a single fault in the circuit, i.e., the applied physical circuit is $EU$ instead of $U$.

\begin{lemma}[$d=2$, flagging errors] \label{lem:ft_flag}
    In this situation, a flag gadget requiring no more than two physical qubits can flag the error $E$, while every possible error in the added gadget is either flagged, its residual error is detectable, or it is logically trivial.
\end{lemma}
\begin{proof}
    Let $P\in \Pn$ be a Pauli operator that anticommutes with $E$. 
    (The flag gadget will catch all errors that anticommute with $P$.)
    Write $Q=U^\dagger PU\in \Pn$ for the backpropagated Pauli operator.
    We choose $P$ to be a tensor product of $I,X,Y,Z$ and a minus sign in $Q$ can be implemented by a $Z$-gate on the first flag qubit or in postprocessing.
    All flag outcomes below include this sign correction.
    We replace the $n$-qubit circuit $U$ with the following $(n+2)$-qubit circuit:
    (i) add two flag qubits initialized in $\ket{+}$, (ii) apply a CZ gate between the two flag qubits, (iii) apply a sequence of controlled-$X$, -$Y$, and -$Z$ gates that implements a controlled-$Q$ gate, where the first flag serves as the control and the code qubits are the targets, (iv) apply the circuit $U$ on the $n$ code qubits, (v) apply a controlled-$P$ gate (decomposed into two-qubit gates) from the first flag qubit to the code qubits, (vi) apply a CZ gate between the two flag qubits, and (vii) read out both flag qubits in the $X$ basis.
    In the absence of errors, the controlled-$Q$ and -$P$ gates cancel and the measurement results are $0$ by construction for both flags, which proves the soundness of the proposed protocol.

    If the single fault occurs that leads to the error $E$ propagating out of the unitary circuit $U$, the first flag qubit will experience a phase kickback through the controlled-$P$ gate, which triggers that flag and the error is detected.

    Let us now show that our flag gadget does not introduce further undetectable errors with a non-trivial logical action.
    We detect errors by discarding the computation if any flag is triggered or if a stabilizer measurement yields a non-zero syndrome.
    First, consider the case where the \emph{final} CZ gate, which is shared by the two flag qubits, has failed. 
    This is unproblematic because errors cannot travel back in time to the data qubits.
    Next, consider the case where the \emph{initial} CZ gate has failed.
    If the first flag-qubit error is $I$ or $Z$, nothing propagates to the data, while if it is $X$ or $Y$, the hook error $Q$ before $U$ is canceled by the hook error $P$ after $U$.
    Finally, consider the case where one of the two-qubit gates in the controlled-$Q$ or -$P$ construction fails.
    There are multiple possibilities: 
    if the Pauli component on the first flag qubit is $I$ or $Z$, the output data error is equivalent to a single-qubit error at the input or output of $U$, depending on whether the faulty gate belongs to the controlled-$Q$ or controlled-$P$ construction.
    Since $U$ implements a logical Clifford gate and the code has distance two, this error is either detectable by the output stabilizers or logically trivial.
    If the flag-qubit component is $X$ or $Y$, the second flag is triggered.
    Faults in the preparation of the flag qubits are either harmless or flip a flag outcome.
    A faulty flag-measurement outcome also triggers rejection.
    Faults in the additional single-qubit gates reduce to the flag-only or data-only cases considered above.
    These are all error sources that need to be considered, which finishes the proof.
\end{proof}

A few comments are in order.
While \cref{lem:ft_flag} gives a general recipe for catching otherwise undetectable errors,
it leaves a lot of room for potential improvement.
Instead of applying the controlled-Pauli operators at the beginning and end of the circuit, they can be propagated to any two points in the circuit, as long as the dangerous error location remains sandwiched between them.
This can save some two-qubit gates; however, one might lose the ability to catch multiple errors at once.
On the other hand, this opens up the option to reuse flag qubits (sometimes even without measuring them).  
Also note that the second flag qubit is often unnecessary if all hook errors are detectable.
In this context, the ordering of the two-qubit controlled-Pauli gates plays a significant role, and fully leveraging this effect remains an open area for further research.

\section{Data availability}

All relevant data are available on GitHub~\cite{htlogicalgates_github}.

\section{Code availability}
For generating \cref{fig:iceberg:bar_chart}, we used Qiskit version 2.4.1.
An open-source implementation of our logical circuit construction algorithm is available on GitHub~\cite{htlogicalgates_github} and can be installed from PyPI via
\inlinecode{pip install htlogicalgates}.

\section{Acknowledgments}

The authors would like to thank 
Lennart Bittel, 
Antonio Anna Mele, 
David Pahl, 
Lukas Pahl, 
Arthur Pesah,
and Armanda Quintavalle for stimulating discussions.
This project has received financial support from the Unitary Foundation, 
the BMFTR 
(QSolid, MUNIQC-Atoms, QuSol, Hybrid++, PASQUOPS),
the DFG (SPP2514 Towards Software for Fault Tolerant Quantum Computing:\,563402549, BoLaCo),
the Munich Quantum Valley, 
Berlin Quantum, 
the Quantum Flagship programs MILLENION and PASQUANS2, 
the DFG (CRC 183), the European Research Council (DebuQC), 
and the European Union's Horizon Europe programme (QuBriC).
This research has been sponsored by IARPA and the Army Research Office, under the Entangled Logical Qubits program, and was accomplished under Cooperative Agreement Number W911NF-23-2-0212. The views and conclusions contained in this document are those of the authors and should not be interpreted as representing the official policies, either expressed or implied, of IARPA, the Army Research Office, or the U.S.~Government. The U.S.~Government is authorized to reproduce and distribute reprints for Government purposes notwithstanding any copyright notation herein.
J.R. is funded by an EPSRC Quantum Career Acceleration Fellowship (grant code: UKRI1224).
J.R. was additionally supported by EPSRC Grants EP/T001062/1 and EP/X026167/1.

\section{Author contributions}

DM and JR conceived the idea. 
EJK and DM developed the theory.
EJK and KL implemented the software. 
EJK worked out and tested the use case examples. 
EJK, KL, and DM created the figures.
EJK, DM, and JE wrote the manuscript with input from all authors.

\section{Competing interests}

The authors declare no competing interests.

\newpage 

\bibliography{main}

\clearpage
\onecolumngrid

\begingroup
\setstretch{1}

\begin{center}
    {\large\bfseries
    Supplementary Material (SM) to\\
    Hardware-tailored logical Clifford circuits for stabilizer codes
    \par}

    \vspace{1em}

    {\normalsize
    Eric J. Kuehnke\textsuperscript{1},
    Kyano Levi\textsuperscript{1},
    Joschka Roffe\textsuperscript{1,2},
    Jens Eisert\textsuperscript{1}
    and Daniel Miller\textsuperscript{1,3}
    \par}

    \vspace{0.4em}

    {\small\itshape
    \textsuperscript{1}Dahlem Center for Complex Quantum Systems,
    Freie Universität Berlin, 14195 Berlin, Germany\par
    \textsuperscript{2}School of Informatics,
    The University of Edinburgh, EH8 9AB Edinburgh, Scotland, UK\par
    \textsuperscript{3}Institute for Theoretical Nanoelectronics (PGI-2),
    Forschungszentrum Jülich, 52428 Jülich, Germany\par}

    \vspace{0.3em}

    {\small (Dated: \today)\par}
\end{center}

\endgroup

\vspace{1.5em}
\par
\twocolumngrid

\begingroup
\renewcommand{\thefootnote}{*}
\footnotetext[0]{%
    \href{mailto:eric.kuehnke@fu-berlin.de}%
         {eric.kuehnke@fu-berlin.de}}
\endgroup

\setcounter{equation}{0}
\setcounter{figure}{0}
\setcounter{table}{0}

\renewcommand{\theequation}{S\arabic{equation}}
\renewcommand{\thefigure}{S\arabic{figure}}
\renewcommand{\thetable}{S\arabic{table}}

\renewcommand{\theHequation}{supp.\arabic{equation}}
\renewcommand{\theHfigure}{supp.\arabic{figure}}
\renewcommand{\theHtable}{supp.\arabic{table}}

\crefname{equation}{Eq.}{Eqs.}
\Crefname{equation}{Equation}{Equations}
\crefname{figure}{Fig.}{Figs.}
\Crefname{figure}{Figure}{Figures}
\crefname{table}{Tab.}{Tabs.}
\Crefname{table}{Table}{Tables}

\setcounter{section}{0}
\setcounter{secnumdepth}{3}
\renewcommand{\thesection}{\Roman{section}}

\renewcommand{\theHsection}{supp.\arabic{section}}

\section{Empirical runtimes of the IQCP solver} \label{app:sec:runtime}

\renewcommand{\arraystretch}{1.5}
\begin{table}
    \centering
    \begin{tabular}{|c||c|c|c|c|c|}
    \hline
         Code & Gate & Length $l$ & CZ count & Found & Optimality\\\hline\hline
          \multirow{2}{*}{$\llbracket 12,2,3\rrbracket$}&\multirow{2}{*}{$\overline{CX_{2,1}}$}&$2$&$11$&$45\unit{min}$&$119\unit{min}$\\
         &&$3$&$9$&$76\unit{h}$&/\\\hline
         \multirow{8}{*}{$\llbracket 8,3,2\rrbracket$}& \multirow{2}{*}{$\overline{H_1}$} & $3$ & $7$ & $8\unit{min}$ & $51\unit{min}$\\
         & & $4$ & $7$ & $21\unit{h}$ & / \\\cline{2-6}
         & \multirow{2}{*}{$\overline{{H}^{\otimes 2}_{1,2}}$} & $3$ & 8 & $86\unit{min}$ & $90\unit{min}$\\
         & & $4$ & 8 & $90\unit{h}$ & $163\unit{h}$ \\\cline{2-6} 
         & \multirow{1}{*}{$\overline{{H}^{\otimes 3}_{1,2,3}}$} & $3$ & $15$ & $20\unit{h}$&/\\\cline{2-6}
         & \multirow{2}{*}{$\overline{{S}_{1}}$}&$1$&$1$&$<1\unit{s}$&$<1\unit{s}$\\
         && $3$ & $1$ & $40\unit{s}$ & $40\unit{s}$ \\\cline{2-6}
         & \multirow{1}{*}{$\overline{{(HS)}_1}$}&$3$&$6$&$61\unit{min}$&$66\unit{min}$\\\hline
         \multirow{1}{*}{$\llbracket 16,6,2\rrbracket$}&\multirow{1}{*}{$\overline{CZ_{1,4}}$}&$3$&$5$&$48\unit{min}$&$119\unit{min}$\\\hline
    \end{tabular}
    \caption{\justifying
    Empirical runtimes of the Gurobi-based IQCP solver used to construct hardware-tailored logical gate implementations for various codes presented in the main text.
    All calculations were carried out on four cores of an Intel Xeon CPU E5-2695 v2 @$2.40\unit{GHz}$ with $20\unit{GB}$ of RAM allocated.
    Note that different logical gates can be constructed in parallel, 
    hence, circuits for a meaningful experiment can often be obtained within hours or days.
    }
    \label{app:tab:runtimes}
\end{table}

In this section, 
we present runtimes of our open-source implementation of the Gurobi-based IQCP solver for 
Eq.~(13) of the main text.
Note that this runtime should be interpreted as the classical preprocessing cost associated with constructing a hardware-tailored circuit implementation of a desired logical Clifford gate for a given stabilizer code.
Making such circuits weakly fault-tolerant in a subsequent step requires applying the techniques outlined in methods section ``Weak fault tolerance with flag gadgets''.
The IQCP solver runtimes for constructing all $720$ Clifford gates for the 
$\llbracket4,2,2\rrbracket$ iceberg code are reported in section ``First use case: $\llbracket4,2,2\rrbracket$ iceberg code'' of the main text.
The runtimes for constructing the remaining circuits from the main text, along with two additional examples reported only here, are shown in \cref{app:tab:runtimes}.
The IQCP solver by Gurobi proceeds in two steps.
First, it constructs \emph{some} feasible point (in the mixed integer/binary parameter space of $B_i, G_i$ variables and slack variables; see Eq.~(13) in the main text)
for the IQCP, 
then it proves optimality (and potentially replaces the initial feasible point by a better one).
The times required for this are shown in \cref{app:tab:runtimes} in the columns ``Found'' and ``Optimality'', respectively, where the latter refers to the total runtime (including the time for constructing the initial feasible point).
For some target circuits, we solve the IQCP for more than one length $l$ of the ansatz circuit $U_{A_l}$ in Eq.~(2) of the main text,
which results in a solution with a different CZ count and a different runtime.

The second row in \cref{app:tab:runtimes} shows that the construction of the circuit in Fig.~2 of the main text---%
which implements a controlled-$X$ gate from qubit 2 to qubit 1 for the $\llbracket12,2,3\rrbracket$ twisted toric code---%
required approximately three days.
By reducing the ansatz length to $l=2$, a similar circuit implementation with two additional CZ gates can be constructed in only 45 minutes, and the optimality (in terms of CZ count and for the given ansatz) is proven in an additional 74 minutes, as shown in the first row of \cref{app:tab:runtimes}.

Next, consider the $\llbracket 8,3,2 \rrbracket$ code as defined in the main text.
The circuits presented in Fig.~3 of the main text were constructed for an ansatz length of $l=3$, 
with solver runtimes ranging from 40 seconds for the logical $S$ gate to 20 hours for the circuit implementing a Hadamard gate on all three logical qubits.
We observe that, for a fixed ansatz length, the solver often runs significantly faster when a circuit implementation with a low CZ count exists.
In the extreme case, the logical $S$ gate requires only a single CZ gate,
and its circuit can be constructed with $l=1$ in under a second.
Note that this simple circuit implementation can also be obtained using the methods of Ref.~\cite{webster_transversal_diagonal_2023}, since the $S$ gate is diagonal in the computational basis.
For the more challenging Hadamard circuits, we also apply our solver with an increased ansatz length of $l=4$.
However, despite the significantly longer runtime, we do not obtain circuits with a lower CZ count.
 
Besides the circuits from Fig.~3 of the main text,
we also construct a logical $\overline{HS}$ gate (not shown) for the $\llbracket 8,3,2 \rrbracket$ code, which  requires approximately 1 hour of preprocessing time and uses six CZ gates with an ansatz length of $l=3$.
Compared to the sequence in which the $S$ gate and then the Hadamard gate  
are applied, the direct implementation of the combined logical $\overline{HS}$ gate saves two CZ gates.
This demonstrates the flexibility of our framework in constructing fully compiled circuits, enabling, for instance, faster access to the $Y$ basis---%
an improvement with important applications~\cite{gidney_inplace_access_2024}.

In the last row of \cref{app:tab:runtimes}, we report a runtime of 48 minutes for constructing a logical CZ gate between two distinct blocks of the $\llbracket8,3,2\rrbracket$ code.
This circuit implementation, available on GitHub~\cite{htlogicalgates_github}, provides another example of the flexibility of our framework in directly tackling the addressability problem of delocalized logical qubits.
Since the two code blocks comprise 16 physical qubits in total, this  substantiates the claim made in the abstract about the scalability of our methods.

To complement the code-specific runtimes in \cref{app:tab:runtimes}, we study the dependence of the computational resource costs on the code size using random stabilizer codes with fixed encoding rate $k/n=1/2$.
For each $n\in\{4,6,8,10,12,14\}$, we sample 100 codes with parameters $\llbracket n,n/2,d\rrbracket$ and $d\geq 2$, 
by rejecting and resampling instances with $d=1$.
For each code, we compile the logical Hadamard gate, $\overline{H}_1$, 
which acts on the first logical qubit.
This is achieved by solving Eq.~(13) of the main text, where we use an ansatz length of $l=1$ and unrestricted connectivity as hyperparameters.
In this way, we assess the dependence of the computational resource costs as a function of the code size.

\Cref{app:fig:runtime_random_codes} shows the empirical distributions of solver runtimes (bottom plot) and recorded peak RAM usage (top plot; left $y$-axis), together with the numbers of IQCP variables and constraints (top plot; right $y$-axis).
We configured Gurobi with \texttt{NodefileStart=0\!.\!5}, 
enabling compressed storage of branch-and-bound tree nodes on disk once their memory consumption exceeds $0.5\unit{GB}$.
Memory usage (excluding disk space) was obtained from Slurm's \texttt{MaxRSS} (maximum resident set size) records with a 30-second accounting interval.
Since memory requirements are unavailable if an IQCP is solved too quickly, the memory distributions describe only the subset of instances---the fraction for which this information is available is reported in the figure (orange annotations over the violins).

We see in \cref{app:fig:runtime_random_codes} that runtimes and recorded memory usage both increase with $n$, with substantial variation between instances.
The runtime distributions also illustrate the additional computational effort required to certify CZ count optimality within the chosen ansatz.
A feasible circuit was found for every sampled instance, while optimality was proven for all instances with $n\leq 12$ and for 95 of the 100 instances with $n=14$ within the configured Gurobi time limit of $1201400\unit{s}\approx 13.9\unit{d}$.

\Cref{app:fig:runtime_random_codes} demonstrates the applicability of our method across randomly sampled codes and provides a reference for its computational resource requirements in the benchmarked setting.
The large variation between instances also shows that code size alone is insufficient to predict the runtime for a particular compilation problem.
Nevertheless, a log-linear fit to the observed minimum total solver runtimes for $n\in \{4,\ldots,14\}$ extrapolates to approximately $1.4\times10^4\unit{s}$ at $n=30$, 
suggesting that even at this size favorable instances may require only several hours of computation.
Since circuit compilation is a one-time preprocessing step,
such runtimes remain practical for repeated experimental use.

\begin{figure}
    \centering
    \includegraphics[width=1\linewidth]{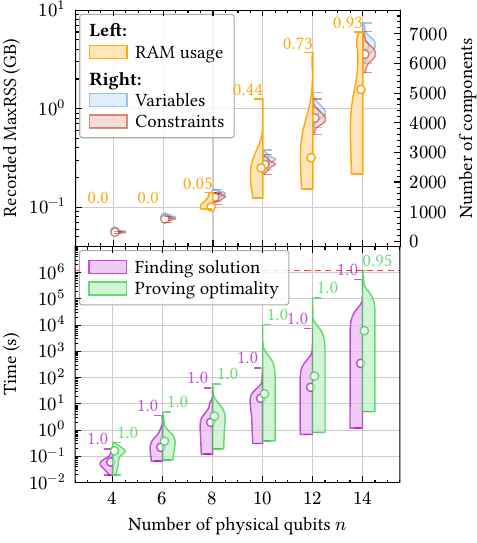}
    \caption{ \justifying
    Empirical computational resource requirements for compiling the logical Hadamard gate $\overline{H}_1$ on 100 random $\llbracket n,n/2,d\rrbracket$ stabilizer codes with $d\geq 2$ as a function of $n$.
Violin shapes show the distributions across instances, horizontal bars indicate the maximum and minimum values, and dots mark the median.
Upper panel: recorded RAM usage (orange, left $y$-axis) and numbers of IQCP variables and constraints (blue and red, right $y$-axis).
Memory values are taken from Slurm's \texttt{MaxRSS}; orange annotations indicate the fraction of jobs for which memory usage was available.
Lower panel: times at which the final reported solutions were found (violet) and total solver runtimes (green), shown on a logarithmic scale (left $y$-axis).
For runs in which optimality was not proven, the corresponding stopping times are shown.
Violet and green annotations indicate the fractions of instances for which a feasible circuit was found and CZ count optimality within the chosen ansatz was proven, respectively.
The dashed red line indicates the configured Gurobi time limit of approximately 14 days.
All calculations were carried out on four cores of an Intel Xeon CPU E5-2695 v2 @$2.40\unit{GHz}$ with $20\unit{GB}$ of RAM allocated.
    }
    \label{app:fig:runtime_random_codes}
\end{figure}

\section{Proof of Theorem 1} \label{app:sec:proof_gauge_freedom}

In this appendix, we prove Theorem 1 from the main text. 
More precisely, we show that
\begin{align} \label{app:eq:definition_freedom_gauge_group}
\mathcal{F}=\{F \in \Sp(\FF_2^{2n}) \ \vert \ F \text{ obeys main text Eq.~(6)}\}
\end{align}
serves as the gauge group for logical Clifford operations whose action is specified only on the logical subspace and may differ outside the code space.
Moreover, we prove that the cardinality of $\mathcal{F}$ is given by the expression in Eq.~(7) in the main text.

\begin{proof}
Our first claim is that $\mathcal{F}$ is a group.
To show this, we write the elements $F \in \mathcal{F}$ in block form as in
\begin{align}
    F = 
        \begin{bmatrix}
         F^{xx} & 0_n \\   
         F^{zx} & F^{zz}    
        \end{bmatrix},
\end{align}
where $0_n\in \FF_2^{n\times n}$ denotes the zero matrix.
The constraints from Eq.~(6) in the main text on the blocks are given by
\begin{align} \label{eq:freedom_matrix_block_xx_and_zz}
    F^{xx} =  \begin{bmatrix}
        \mathbbm 1_{k} & \ast &\cdots & \ast \\
        0 & \ast & \cdots & \ast \\
        \vdots & \vdots & \ddots & \vdots \\
        0 & \ast & \cdots & \ast \\
    \end{bmatrix},
    \hspace{5mm}
F^{zz} &=  \begin{bmatrix}
        \mathbbm 1_{k} & 0 &\cdots & 0 \\
        \ast & \ast & \cdots & \ast \\
        \vdots & \vdots & \ddots & \vdots \\
        \ast & \ast & \cdots & \ast \\
    \end{bmatrix}, 
    \\
\text{and } \hspace{2mm}   
    F^{zx} &=  \begin{bmatrix}
         0_{k} & \ast &\cdots & \ast \\
        \ast & \ast & \cdots & \ast \\
        \vdots & \vdots & \ddots & \vdots \\
        \ast & \ast & \cdots & \ast \\
    \end{bmatrix}.
    \label{eq:freedom_matrix_block_zx}
\end{align}
Clearly, $F = \mathbbm1_{2n}$ fulfills these constraints, proving $\mathbbm1_{2n} \in \mathcal{F}$. 
Taking the product of two matrices $F,\tilde F \in \mathcal{F}$ results in
\begin{align}
    F\tilde F = \begin{bmatrix}
        F^{xx} \tilde{F}^{xx} & 0_n \\
        F^{zx} \tilde{F}^{xx}  + F^{zz} \tilde{F}^{zx}  & F^{zz} \tilde{F}^{zz} 
    \end{bmatrix}.
\end{align}
It is straightforward to verify that $F^{xx} \tilde{F}^{xx}$ and
$F^{zz} \tilde{F}^{zz}$ inherit the constraints of \cref{eq:freedom_matrix_block_xx_and_zz}.
Similarly, both $F^{zx} \tilde{F}^{xx} $ and $ F^{zz} \tilde{F}^{zx}$ 
obey the constraints of \cref{eq:freedom_matrix_block_zx}, which proves $F\tilde F\in \mathcal{F}$.
This shows that $\mathcal{F}$ is closed under taking products. 
Hence, it is also closed under taking inverses because $\mathcal{F} \subset \FF_2^{2n\times 2n}$ is clearly finite.
This proves that $\mathcal{F}$ is a group.

As mentioned in the main text, $\mathcal{F}$ can be understood as the subgroup of non-Pauli Clifford stabilizers of the code space $\mathcal{L}$ of an $\llbracket n,k,d\rrbracket$ code.
Next we will show that  $\mathcal{F}$ is in bijection with the different choices of physical Clifford operators (modulo Paulis and phases) that implement a given logical Clifford gate.
In that sense, the elements of $\mathcal{F}$ correspond to different gauges of logical Clifford operators.
To distinguish it from the concept of \emph{gauge groups} in subsystem QECCs~\cite{lidar_brun_quantum_error_2013, vuillot_code_deformation_2019, li_2d_compass_2019},
we refer to $\mathcal{F}$ as the \emph{freedom gauge group}.

We need to show that $U_{EC'FE^{-1}}$ implements $U_C\in \mathcal{C}^k$ on the logical level whenever $F\in \mathcal{F}$ and, conversely, 
that every $n$-qubit Clifford operation that does so is of the form  $U_{EC'FE^{-1}}$ for some $F\in \mathcal{F}$.
For both statements we will make use of the fact that the encoding circuit $U_E$ maps Pauli-$Z$ operators on qubits $k+1$ to $n$ to stabilizers $S\in \langle S_1, \ldots, S_{n-k}\rangle $ and arbitrary Pauli operators on qubits $1$ to $k$ to logical Pauli operators $\bar P \in \langle \bar X_1,\bar Z_1,\ldots,\bar X_k,\bar Z_k \rangle$.
To make this more precise, we introduce binary vectors $\mathbf{s}_j, \mathbf{s}'_j, \mathbf{x}_i, \mathbf{x}'_i, \mathbf{z}_i, \mathbf{z}'_i \in \FF_2 ^n$ such that $S_{j-k} \propto X^{\mathbf{s}_j} Z ^{\mathbf{s}'_j}$, for all $j\in \{k+1,\ldots, n\}$ and $\bar X_i \propto X^{\mathbf{x}_i} Z ^{\mathbf{x}'_i}$, $\bar Z_i \propto X^{\mathbf{z}_i} Z ^{\mathbf{z}'_i}$ for all $i\in\{1,\ldots, k\}$, see Eq.~(20) in the main text.
Then, we have 
\begin{align}
    E 
    \begin{bmatrix}
    \mathbf{e}_i \\ 0    
    \end{bmatrix}
    =
    \begin{bmatrix}
        \mathbf{x}_i \\ \mathbf{x}'_i
    \end{bmatrix},
    \hspace{2mm}
   E 
    \begin{bmatrix}
    0 \\ \mathbf{e}_i     
    \end{bmatrix}
   & =
    \begin{bmatrix}
        \mathbf{z}_i \\ \mathbf{z}'_i
    \end{bmatrix},
    \label{app:eq:unencoded_pauli_operators}
  \\\label{app:eq:unencoded_stabilizer_operators}
  \hspace{2mm}
    \text{and}
    \hspace{2mm}
    E 
    \begin{bmatrix}
    0 \\ \mathbf{e}_j     
    \end{bmatrix}
    &=
    \begin{bmatrix}
        \mathbf{s}_j \\ \mathbf{s}'_j
    \end{bmatrix},
\end{align}
where $\mathbf{e}_l = (\delta _ {l,l'})_{l'=1}^{n} \in \FF_2^n$ 
denotes a standard basis vector.

First, assume we have a Clifford operation $U_M\in \Cn$ that is represented by $M = EC'FE^{-1}$ for some $F \in \mathcal{F}$.
We have to show that $U_M$ acts as $C$ on the logical 
level.
For every $j \in \{k+1,\ldots, n\}$, we find that $U_M S_{j-k} U_M^\dagger$ is represented by
\begin{align}
M\begin{bmatrix}
    \mathbf{s}_j \\ \mathbf{s}'_j
\end{bmatrix}
=&E\sum_{l=k+1}^{n} F_{l+n,j+n}\begin{bmatrix}
    0 \\ \mathbf{e}_l
\end{bmatrix},
\end{align}
where we have used \cref{app:eq:unencoded_stabilizer_operators} and the fact that $F^{xz}=0$, which holds by definition of  $F\in \mathcal{F}$.
Again applying \cref{app:eq:unencoded_stabilizer_operators}, we further find
\begin{align}
M\begin{bmatrix}
    \mathbf{s}_j \\ \mathbf{s}'_j
\end{bmatrix}
=&\sum_{l=k+1}^{n} F_{l+n,j+n}\begin{bmatrix}
    \mathbf{s}_l \\ \mathbf{s}'_l
\end{bmatrix},
\end{align}
which, after a suitable Pauli-frame correction to $U_M$, implies $U_M S_{j-k} U_M^\dagger\in \mathcal{S}$. 
Therefore, Lemma~3 from methods section ``Characterization of logical Clifford gates'' applies,
and shows that $U_M$ is a logical operator.
To determine its logical action, let us first compute the vectors that represent $U_M \overline{X_i}U_M^\dagger $ for all $i\in \{1,\ldots, k\}$, that is 
\begin{align}
    M\begin{bmatrix}
    \mathbf{x}_i \\ \mathbf{x}'_i
\end{bmatrix}&=EC'\begin{bmatrix}
    \mathbf{e}_i \\ 0
\end{bmatrix}+\sum_{l=k+1}^{n} F_{l+n,i}\begin{bmatrix}
    \mathbf{s}_l \\ \mathbf{s}'_l
\end{bmatrix}.
\label{app:eq:transformed_x}
\end{align}
The first term in \cref{app:eq:transformed_x} represents the logical Pauli operator into which the $i$-th Pauli-$X$ operator is transformed,
while the second term reflects that this mapping is defined modulo stabilizers.
Similarly, we find 
\begin{align}
U_M \overline{Z_i}U_M^\dagger = \overline{U_CZ_iU_C^\dagger}S'_i
\end{align}
for some $S'_i\in \mathcal{S}$. 
Therefore, Lemma~4 from methods section ``Characterization of logical Clifford gates'' applies, and we have shown that $U_M$ acts as $U_C\in \mathcal{C}^k$ on the logical level.

Next, we show the converse that every Clifford operation $U_M\in \Cn$ that implements $U_C\in \mathcal{C}^k$ on the logical level can be written as $M={EC'FE^{-1}}$ for some freedom matrix $F\in \mathcal{F}$.
Our strategy is to translate the constraints on $M$ imposed by Lemmata~3 and 4 from methods section ``Characterization of logical Clifford gates'' into the constraints on $F$ that are shown in Eq.~(6) in the main text.
To this end, we define the freedom matrix $F = C'^{-1}E^{-1}ME$.
To prove $F\in \mathcal{F}$, we first apply $F$ to the unit vector corresponding to the unencoded $i$-th logical Pauli-$X$ operator from \cref{app:eq:unencoded_pauli_operators}.
By \cref{app:eq:unencoded_pauli_operators}, this yields
\begin{align}
    F\begin{bmatrix} \mathbf{e}_i \\ 0\end{bmatrix}&=
C'^{-1}E^{-1}M\begin{bmatrix}
    \mathbf{x}_i \\ \mathbf{x}'_i
\end{bmatrix}.
\end{align}
Next, we apply Lemma~4 from methods section ``Characterization of logical Clifford gates'' to arrive at
\begin{align}
F\begin{bmatrix} \mathbf{e}_i \\ 0\end{bmatrix}
&=\begin{bmatrix}
    \mathbf{e}_i \\ 0
\end{bmatrix}+\sum_{l=k+1}^n F_{l+n,i} \begin{bmatrix}
    0 \\ \mathbf{e}_l
\end{bmatrix} \,,
\end{align}
which proves the restrictions on $F$ shown in Eq.~(6) in the main text for columns $1$ to $k$.
Similarly, by repeating the calculation for the logical Pauli-$Z$ operators, we obtain the corresponding restrictions for columns $n+1$ to $n+k$.
Finally, we analyze the constraints on $F$ that are imposed by how $U_M$ is allowed to transform stabilizer operators.  
By Lemma~3 from methods section ``Characterization of logical Clifford gates'', we have
\begin{align}
    F\begin{bmatrix}
    0 \\ \mathbf{e}_j
\end{bmatrix}&=C'^{-1}E^{-1}M\begin{bmatrix}
    \mathbf{s}_j \\ \mathbf{s}'_j
\end{bmatrix}\\
&=\sum_{l=k+1}^n F_{l+n,j+n} \begin{bmatrix}
    0 \\ \mathbf{e}_l 
\end{bmatrix},
\nonumber
\end{align}
which proves that also columns $n+k+1$ to $2n$ have to be of the form given in Eq.~(6) in the main text.
Since there are no constraints, besides $F\in \Sp(\FF_2^{2n})$, on columns $k+1$ to $n$, this finishes the proof of $F\in \mathcal{F}$.

Finally, let us compute the order $\vert \mathcal{F}\vert $ of the freedom gauge group.
Because every freedom matrix $F\in \mathcal{F}$ is invertible, the same must be true about the block matrix $F^{zz}$.
By \cref{eq:freedom_matrix_block_xx_and_zz}, 
this is the case iff the submatrix of size
$(n-k)\times(n-k)$ in the bottom right of $F^{zz}$ is invertible. 
Besides this, there are no further constraints on the column vectors of $F^{zz}$.
Thus, there are
\begin{equation}
| \FF_2^{n-k}|^k \times   |\operatorname{GL}(\FF_2^{n-k})|=2^{k(n-k)}\prod_{i=0}^{n-k-1}(2^{n-k}-2^i)
\end{equation}
possible choices for $F^{zz}$.
Next, we write out the condition $F\in \Sp(\FF_2^{2n})$, i.e.,
$     F^T \left[\begin{smallmatrix}
        0&\mathbbm 1\\
        \mathbbm 1&0
    \end{smallmatrix}\right] 
    F =
    \left[\begin{smallmatrix}
        0&\mathbbm 1\\
        \mathbbm 1&0
    \end{smallmatrix}\right]$,
which yields
\begin{align}    \label{app:eq:freedom_matrix_constraint2}
    (F^{xx})^TF^{zz} &= \mathbbm 1  \text{  \ \ and}\\
    \label{app:eq:freedom_matrix_constraint}
    \left((F^{xx})^TF^{zx}\right)^T &= (F^{xx})^TF^{zx}.
\end{align}
By \cref{app:eq:freedom_matrix_constraint2}, 
the submatrix $F^{xx} = ((F^{zz})^T)^{-1}$ is uniquely determined through the choice of $F^{zz}$.
\Cref{app:eq:freedom_matrix_constraint} means that $(F^{xx})^TF^{zx}$ must be a symmetric matrix.
Since we can regard $(F^{xx})^T$ as a bijective map, 
the numbers of allowed choices for $(F^{xx})^TF^{zx}$ and $F^{zx}$ are identical.
There are in total $2^{n(n+1)/2}$ symmetric binary $n\times n$-matrices.
However, not all of them are allowed by Eq.~(6) in the main text.
Rows $1$ to $k$ of $F^{zz}$ map columns $1$ to $k$ of $(F^{xx})^TF^{zx}$ to zero in $F^{zx}$, since $F^{zz}[(F^{xx})^TF^{zx}]=F^{zx}$.
This reduces the number of free variables of $(F^{xx})^TF^{zx}$ by $k(k+1)/2$ for a given choice of $F^{zz}$ and $F^{xx}$.
In total, this shows that the order of the freedom gauge group $\mathcal{F}$ is indeed given by the expression in Eq.~(7) in the main text,
which finishes the proof of Theorem 1.
\end{proof}

This proof contains an explicit (albeit somewhat opaque) enumeration of all elements in the freedom gauge group $\mathcal{F}$.
To better understand the role of  $\mathcal{F}$, 
we rearrange Eq.~(7) in the main text into 
\begin{align}
    \begin{split}
    |\mathcal{F}|=&\underbrace{2^{2k(n-k)}}_{\text{(1)}}
    \underbrace{2^{(n-k)(n-k+1)/2}}_{\text{(2)}}
    \underbrace{\prod_{i=0}^{n-k-1}\left(2^{n-k}-2^{i}\right)}_{\text{(3)}}.\label{app:eq:freedom_number_sorted}
\end{split}
\end{align}
Factor (3) in \cref{app:eq:freedom_number_sorted} is the number of ways in which the $n-k$ stabilizer generators can be mapped to a different choice of stabilizer generators.
Similarly, factor (1) is the number of ways to correctly transform $2k$ logical Pauli operators modulo stabilizers,
while factor (2) characterizes the additional freedom provided by the transformation of Pauli errors.

\section{Circuit-level noise simulation} \label{app:sec:simulation}

In this appendix, we provide details about the circuit-level noise simulations that we carried out to produce Fig.~4 in the main text. 

The goal of these simulations is twofold: to provide a numerical sanity check that our logical Hadamard circuits for the $\llbracket 8,3,2 \rrbracket$ code are indeed weakly fault-tolerant (FT), and to compare their performance against an existing protocol~\cite{wang_fault_tolerant_2024}.
All simulations were carried out with Stim~\cite{gidney_stim_2021},
using a circuit-level noise model in which all error probabilities (of gates, measurements, resets, etc.) are proportional to a single physical error parameter called $p$.
This parameter serves as the $x$-axis of Fig.~4 in the main text.
Our noise model is inspired by the one in Tab.~3 of Ref.~\cite{gidney_inplace_access_2024}.
However,
we find it convenient to slightly modify it as follows.
Besides controlled-$Z$ gates,
we also include controlled-$X$ and -$Y$ gates in the set of available gates.
For these gates, we assume the same two-qubit depolarizing noise as for the controlled-$Z$ gate.
Besides measurements and resets in the $Z$-basis, 
we also include the same operations in the $X$-basis.
Here, we replace bit errors with phase errors, while keeping the same error rates.
Note that the conclusions drawn in the main text  are not restricted to this specific choice of noise model.

Before running a circuit-level error simulation, we have to specify a circuit that includes (1) weakly FT state preparation, 
(2) weakly FT logical gates,
(3) weakly FT stabilizer measurements,
and (4) weakly FT measurement of logical Paulis.

Let us explain how we select these components to study logical Hadamard gates for the $\llbracket8,3,2\rrbracket$ code.
First, we need weakly FT state preparation circuits.
To prepare the logical state vector $\ket{\overline{0,0,0}}$, we can initialize every physical qubit in $\ket{0}$ and perform a weakly FT measurement (see below) of the only $X$-type stabilizer operator, $X^{\otimes 8}$.
If instead we wish to prepare $\ket{\overline{+, +, +}}$, we use the circuit shown in Fig.~3 of Ref.~\cite{menendez_implementing_fault_2024}.
These circuits are also available on GitHub~\footnote{See \href{https://github.com/erkue/htlogicalgates/blob/main/examples/color832_prepare_000.txt}{{color832\_prepare\_000.txt}} and \href{https://github.com/erkue/htlogicalgates/blob/main/examples/color832_prepare_ppp.txt}{{color832\_prepare\_ppp.txt}} in \href{https://github.com/erkue/htlogicalgates/tree/main/examples}{https://github.com/erkue/htlogicalgates/tree/main/examples}.}.

Second, we need FT logical gates.
For this, we work either with our new weakly FT gates constructed in the main text or with the teleportation-based constructions from Ref.~\cite{wang_fault_tolerant_2024}, which are also available on GitHub~\footnote{See files \href{https://github.com/erkue/htlogicalgates/tree/main/examples/color832_wang_H0.txt}{color832\_wang\_H0.txt} and \href{https://github.com/erkue/htlogicalgates/tree/main/examples/color832_wang_H01.txt}{color832\_wang\_H01.txt} in  \href{https://github.com/erkue/htlogicalgates/tree/main/examples}{https://github.com/erkue/htlogicalgates/tree/main/examples}.}.
Recall from Definition~1 in methods section ``Weak fault tolerance with flag gadgets''
that we define a logical gate for a code with distance $d=2$ to be weakly FT if every single error is flagged, its residual error is detectable, or it has a trivial logical action.
A detectable error, however, is not the same as a detected error.
To effectively remove error mechanisms from a circuit, 
we need to detect the errors by performing a round of stabilizer measurements.
For the $\llbracket 8,3,2\rrbracket$ code, 
this can be ensured by employing a flag construction similar to that in Lemma~5 in methods section ``Weak fault tolerance with flag gadgets'', with the roles of flags 1 and 2 taken by the auxiliary qubit and the flag, respectively.
Finally, we can implement the weakly FT measurement of logical Pauli operators by performing a round of stabilizer measurements, followed by reading out physical qubits in the basis corresponding to any $X$- or $Z$-type representative of the logical operator.
Here, we save resources by inferring the syndromes of stabilizers whose eigenvalues can be reconstructed as parities of the physical measurement outcomes, rather than measuring those stabilizers with a weakly FT stabilizer extraction circuit, using flags.
For example, to perform a weakly FT measurement of $\overline{Z_1}$, $\overline{Z_2}$, and $\overline{Z_3}$, 
we execute a flag-based weakly FT stabilizer measurement for $X^{\otimes 8}$ 
before we read out all eight physical qubits in the computational basis. 
For the weakly FT measurement of $\overline{X_1}$, $\overline{X_2}$, and $\overline{X_3}$, we could, in principle, proceed similarly.
However, instead we execute the time-reversed circuit of the weakly FT state preparation circuit for $\ket{\overline{\mathord{+},\mathord{+},\mathord{+}}}$ from Ref.~\cite{menendez_implementing_fault_2024}.
Here, weak fault tolerance follows from the fact that every fault occurring during the reversed circuit propagates to at most a detectable error.
These circuits are also available on GitHub~\footnote{See files \href{https://github.com/erkue/htlogicalgates/blob/main/examples/color832_read_000.txt}{{color832\_read\_000.txt}} and \href{https://github.com/erkue/htlogicalgates/blob/main/examples/color832_read_ppp.txt}{{color832\_read\_ppp.txt}} in \href{https://github.com/erkue/htlogicalgates/tree/main/examples}{https://github.com/erkue/htlogicalgates/tree/main/examples}.}.

Finally, we combine these building blocks into weakly FT circuits that map $\ket{\overline{0,0,0}}$ to $\ket{\overline{\mathord{+},\mathord{+},\mathord{+}}}$ and vice versa.
All gates, measurements (including those of the flag qubits in the stabilizer measurements), and resets are subject to the noise model; no element is treated as noise-free.
The circuits used in the simulations are available online on GitHub~\cite{htlogicalgates_github}.

 \end{document}